\documentclass[11pt, letterpaper]{article}

\usepackage{pbg_template}
\usepackage{parskip}
\usepackage{float}
\usepackage{longtable}
\usepackage{array}
\usepackage{caption}
\usepackage[utf8]{inputenc} 
\usepackage[T1]{fontenc}    
\usepackage{hyperref}       
\usepackage{url}            
\usepackage{booktabs}       
\usepackage{amsfonts}       
\usepackage{nicefrac}       
\usepackage{microtype}      
\usepackage{xcolor}         
\usepackage{graphicx}       
\usepackage{subcaption}     

\usepackage[numbers,sort&compress]{natbib}
\usepackage{amsmath}
\usepackage{amssymb}
\usepackage{mathtools}
\usepackage{amsthm}

\usepackage{algorithm}
\usepackage{algpseudocode}
\usepackage{xcolor}

\usepackage[capitalize,noabbrev]{cleveref}

\theoremstyle{plain}
\newtheorem{theorem}{Theorem}[section]

\definecolor{pennred}{RGB}{153,0,0}
\definecolor{pennblue}{RGB}{0, 85, 160}

\hypersetup{
    colorlinks=true,
    linkcolor=pennred,
    citecolor=pennred,
}

\theoremstyle{definition}

\theoremstyle{remark}

\usepackage{booktabs}
\usepackage{tabularx}
\usepackage{longtable}
\usepackage{multirow}
\usepackage{array}
\usepackage[table]{xcolor}

\tcbuselibrary{breakable}
\definecolor{thmbg}{HTML}{fddbc7}

\title{SF-Cluster: Frustration-Guided MSA Subsampling for Alternative Protein Conformation Recovery}

\author{Hanqun Cao,\textsuperscript{1,*} Zijun Gao,\textsuperscript{1,*} Chunbin Gu,\textsuperscript{1} Ge Liu,\textsuperscript{2} Pheng Ann Heng,\textsuperscript{1},    Pranam Chatterjee\textsuperscript{3,4,\dag}

    \vspace{1em}
    \normalfont \small
    \textsuperscript{1}Department of Computer Science and Engineering, The Chinese University of Hong Kong\\
    \textsuperscript{2}Department of Computer Science, University of Illinois at Urbana-Champaign\\
    \textsuperscript{3}Department of Bioengineering, University of Pennsylvania\\
    \textsuperscript{4}Department of Computer and Information Science, University of Pennsylvania\\
    \vspace{0.5em}
    \textit{\textsuperscript{*}Equal contribution}
    \vspace{0.5em}\\
    \textbf{Correspondence:} \href{mailto:pranam@seas.upenn.edu}{\texttt{pranam@upenn.edu}} \\
}

\begin{document}

\maketitle

\begin{abstract}
Deep-learning structure predictors are sensitive to their multiple sequence alignment (MSA) input, making MSA subsampling a practical route to recovering alternative conformations. Existing approaches such as AF-Cluster operate in sequence space, providing limited control over which conformational basin is sampled. We introduce \textbf{SF-Cluster}, which subsamples MSAs using patterns of predicted local energetic frustration, a representation largely independent of sequence similarity. Across a benchmark of 48 cases spanning fold-switching, allosteric, oligomerization-coupled, and intrinsically disordered systems, and using an AF-Cluster-style dual-reference RMSD criterion, SF-Cluster improves target-state recovery of the alternative conformation over AF-Cluster across the two-state classes, with the largest improvement observed for allosteric systems (+15.5 percentage points). The selected MSAs transfer to an architecturally distinct predictor, indicating that the conformational signal resides in MSA composition. Mechanistically, matched-depth controls show that this recovery advantage is largely explained by the effective depth of the selected subsets, which frustration-pattern selection reliably reaches. At the same time, highly frustrated residues are enriched at sites supported by deep mutational scanning and NMR two-state exchange, and frustration covariation is enriched at state-switching contacts while remaining distinct from coevolutionary coupling. Together, these results identify frustration patterns as a transferable representation for conformational prediction and position MSA subsampling as a representation-guided reweighting problem.
\end{abstract}

\begin{figure}[H]
\centering
\includegraphics[width=\textwidth]{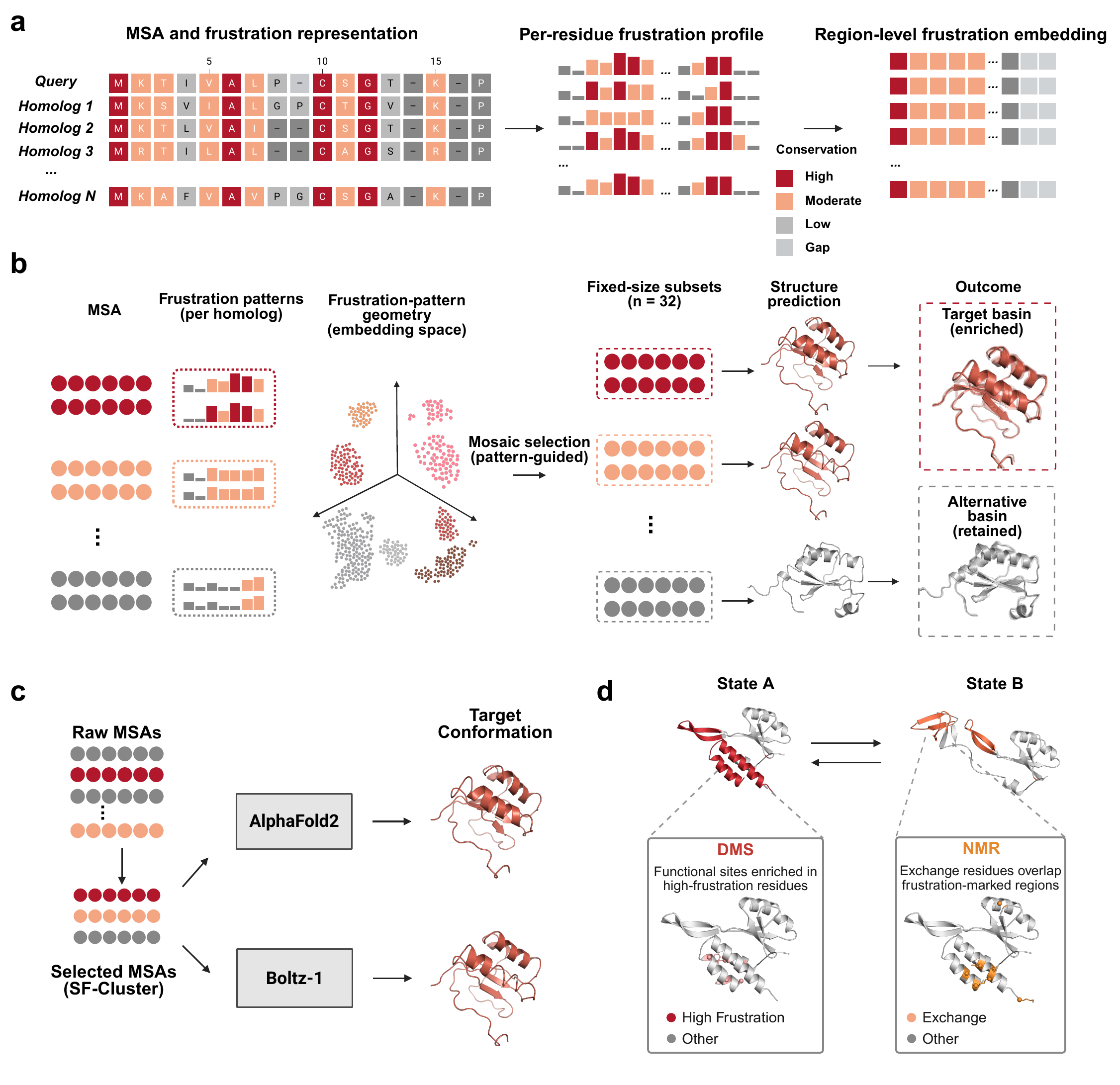}
\caption{\textbf{Overview of SF-Cluster.}
\textbf{(a)} For each homolog in the MSA, a per-residue frustration profile is computed and compressed into a region-level embedding (frustration High to Low; gaps grey).
\textbf{(b)} Homologs are positioned by their geometry in frustration-pattern space, and a single mosaic selection draws fixed-size subsets ($n=32$) spanning diverse modes, intended to enrich the target basin while retaining the alternative (schematic).
\textbf{(c)} The selected MSAs recover the target conformation through two architecturally distinct predictors, AlphaFold2 and Boltz-1, placing the conformational signal in MSA composition rather than one model.
\textbf{(d)} The frustration signal is anchored to independent data: high-frustration residues coincide with deep mutational scanning (DMS) functional sites, and frustration-marked regions overlap NMR two-state exchange residues.}
\label{fig:framework}
\end{figure}

\section{Introduction}

Conformational heterogeneity is central to protein function. Fold-switching proteins, allosteric regulators, and receptor systems depend on transitions between structurally distinct states \cite{porter2018extant, kim2021functional, cao2026tdb}. Despite the success of AlphaFold2~\citep{jumper2021highly,porter2018extant} in single-structure prediction, recovering alternative conformations from sequence alone remains difficult. Predictions are typically biased toward dominant states, and directing inference toward a specific conformational basin generally requires modifications to the standard prediction pipeline~\citep{chakravarty2022alphafold2, chakravarty2024alphafold}.

MSA subsampling has emerged as a practical approach for increasing conformational diversity. Because AlphaFold2 is sensitive to its MSA input, changes to the MSA can alter the conformational basin sampled during prediction~\citep{del2022sampling, stein2022speach_af}. AF-Cluster demonstrated this effect by partitioning an MSA according to sequence similarity and running structure prediction on each subset, showing that cluster-specific MSAs can recover multiple conformations for fold-switching proteins~\citep{wayment2024predicting}. Subsequent methods have refined this strategy through sequence purification and stochastic inference~\citep{xing2025leveraging, bryant2024structure, lee2025large, li2026disentangling}. However, these approaches operate in sequence or perturbation space, where sequence similarity provides limited guidance for selecting MSA subsets associated with a desired conformational state.

Local energetic frustration provides an alternative representation of conformational preference. Frustration quantifies how well residue interactions are satisfied relative to alternative configurations and has been linked to conformational strain, switching behavior, and alternative fold preferences~\citep{leusch2026frustrai,ferreiro2014frustration}. Frustration patterns across homologous sequences therefore contain information that is complementary to sequence similarity. 

Here, we introduce \textbf{SF-Cluster}, an MSA subsampling framework that uses frustration-derived representations to guide state-directed conformational sampling together with a coverage-aware refinement procedure that preserves low-frequency conformations that would otherwise be removed by confidence-based ranking. Across fold-switching benchmarks, SF-Cluster improves recovery of target conformations relative to AF-Cluster. Successful MSA subsets are organized by protein-specific frustration patterns whose geometry differs across systems, indicating that conformational signals encoded within an MSA occupy preferred directions in frustration space. When the input MSA contains only a single structural basin, frustration-guided subsampling does not recover alternative conformations, defining a practical limit of what MSA subsampling can achieve. Together, these results suggest a view of MSA subsampling in which sequence-space exploration, frustration-guided focusing, and refinement of low-frequency states play complementary roles in recovering conformational diversity.

\section{Results}
\label{sec:results}

\subsection{SF-Cluster improves AF-Cluster-style target-state recovery over sequence-space clustering}
\label{sec:recovery}

Sequence similarity provides limited information about conformational state, so MSA clusters built in sequence space, the basis of AF-Cluster~\cite{wayment2024predicting}, often contain homologs associated with different structural basins. We therefore selected MSA subsets according to their geometry in a space of predicted local energetic frustration (Methods). Throughout, SF-Cluster denotes a single mosaic selection on raw per-residue frustration features applied uniformly across all proteins without per-target tuning. We evaluated two prediction tasks separately and never pooled them: two-state conformational recovery on fold-switching, allosteric, and oligomerization-coupled systems, and single-state fold recovery on intrinsically disordered regions. The primary benchmark compares SF-Cluster and AF-Cluster on the two-state task using a dual-reference criterion that requires a prediction to lie within $3$~\AA{} of the target reference, be closer to the target than to the dominant reference, and achieve a mean pLDDT of at least $70$ (Methods). Common-core and loose-RMSD criteria are reported only as structural-proximity sensitivity analyses. During curation we identified and removed two benchmark entries whose paired reference structures corresponded to different proteins (Methods).

SF-Cluster improved target-state recovery over AF-Cluster in every two-state class (Fig.~\ref{fig:recovery}a; \Cref{tab:controls}; SI~\ref{si:impl:controls}). Recovery on allosteric systems increased from $0.273$ to $0.428$, a gain of $15.5$ percentage points and the largest improvement observed across the benchmark. Adenylate kinase provides a representative example of this class, where SF-Cluster recovered the alternative state while AF-Cluster did not~\citep{ferreiro2011role}. Recovery on fold-switching proteins increased from $0.033$ to $0.101$, although this remained the most challenging category. Oligomerization-coupled systems showed the same trend, increasing from $0.008$ to $0.114$, although many of these systems are difficult to evaluate cleanly as two-state recovery problems.

The same ordering was preserved across additional recovery metrics. Under the target-proximity component of the criterion, $39\%$ of SF-Cluster predictions fell within $3$~\AA{} of the target structure compared with $29\%$ for AF-Cluster (Fig.~\ref{fig:recovery}b). SF-Cluster also reduced the best target RMSD achieved per case (Fig.~\ref{fig:recovery}c), increased the fraction of predictions assigned to the target fold by TM-score (Fig.~\ref{fig:recovery}d), and improved per-case recovery rates across the benchmark (Fig.~\ref{fig:recovery}e).

The improvement was consistent across all-$\alpha$, mixed $\alpha$-$\beta$, and all-$\beta$ CATH architectures~\citep{orengo1997cath}, and was generally larger for longer proteins, where sequence-space clustering produces the least controlled subsets. Per-class and per-length analyses are provided in the Supplementary Information. Recovery remained stable across subset sizes near the deployed value of $32$, and per-strategy analyses on the development systems are reported in SI~\ref{app:strategies}.

\begin{figure}[h!]
\centering
\includegraphics[width=0.7\linewidth]{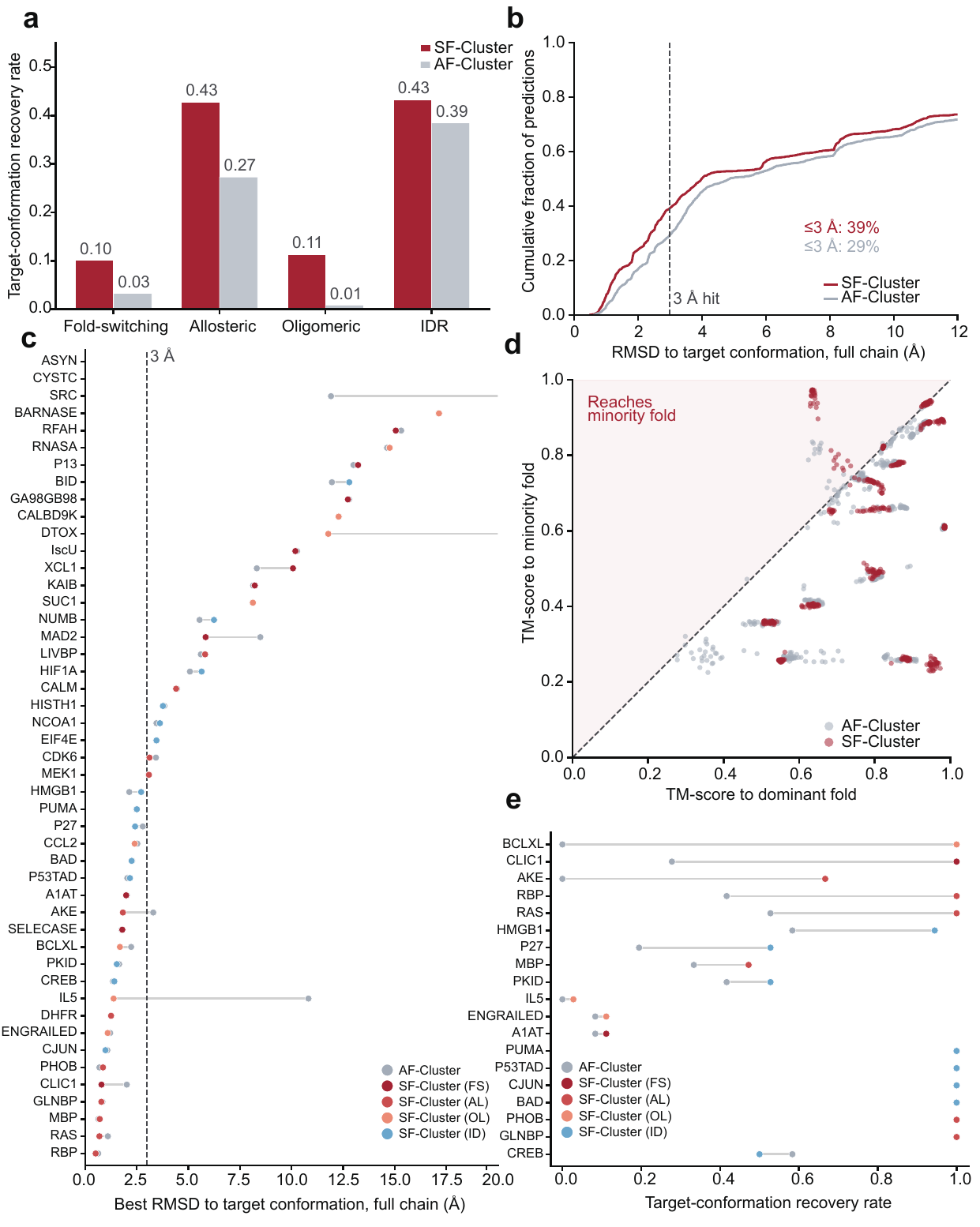}
\caption{\textbf{SF-Cluster improves AF-Cluster-style target-state recovery of the alternative conformation.} Two-state recovery was scored using the AF-Cluster-style dual-reference criterion, requiring a full-chain C$\alpha$ RMSD to the target reference of at most $3$~\AA{}, a lower RMSD to the target than to the dominant reference, and a mean pLDDT of at least $70$ (Methods). \textbf{(a)} Target-state recovery rate by class on the cleaned benchmark. Fold-switching, allosteric, and oligomerization-coupled systems were evaluated using the dual-reference criterion, while intrinsically disordered regions (IDRs) were evaluated using the single-state ordered-form criterion. \textbf{(b)} Cumulative distribution of full-chain C$\alpha$ RMSD to the target reference. Thirty-nine percent of SF-Cluster predictions and $29\%$ of AF-Cluster predictions fall within $3$~\AA{} of the target structure (dashed line). \textbf{(c)} Best full-chain C$\alpha$ RMSD to the target reference achieved for each benchmark case. Dashed line, $3$~\AA{} threshold. \textbf{(d)} TM-score to the target fold versus the dominant fold for individual predictions. The shaded region denotes predictions assigned to the target fold by TM-score. \textbf{(e)} Per-case target-state recovery rates under the dual-reference criterion for AF-Cluster and SF-Cluster, coloured by benchmark class.}
\label{fig:recovery}
\end{figure}

\subsection{The conformational signal resides in the MSA and transfers to a second predictor}
\label{sec:transfer}

The recovery results above were obtained with AlphaFold2, leaving open whether SF-Cluster identifies conformational information encoded in the MSA or exploits properties specific to a single predictor. To distinguish these possibilities, we supplied SF-Cluster-selected MSAs to Boltz-1~\cite{boltz1}, a structure predictor with an architecture distinct from AlphaFold2, on a ten-protein subset spanning the benchmark classes (Fig.~\ref{fig:transfer}a and Fig~\ref{fig:transfer}b). Boltz-1 recovered none of the target conformations from single-sequence input, but recovered four of ten when provided with SF-Cluster-selected MSAs. Because the only quantity transferred between the two predictors is the MSA itself, these results indicate that the conformational signal resides in MSA composition and is not specific to AlphaFold2. They also provide independent support that recovery of alternative conformations is not tied to a particular prediction architecture~\cite{chakravarty2024alphafold}.

A depth-matched random control also transferred to Boltz-1, recovering three of ten target conformations compared with four of ten for SF-Cluster-selected MSAs (McNemar $p = 1.0$). Together, these results indicate that conformational information is encoded in the selected alignments and remains accessible across prediction architectures. The similarity between SF-Cluster and the depth-matched control further suggests that alignment depth contributes substantially to this transferability.

\begin{figure}[h]
\centering
\includegraphics[width=\linewidth]{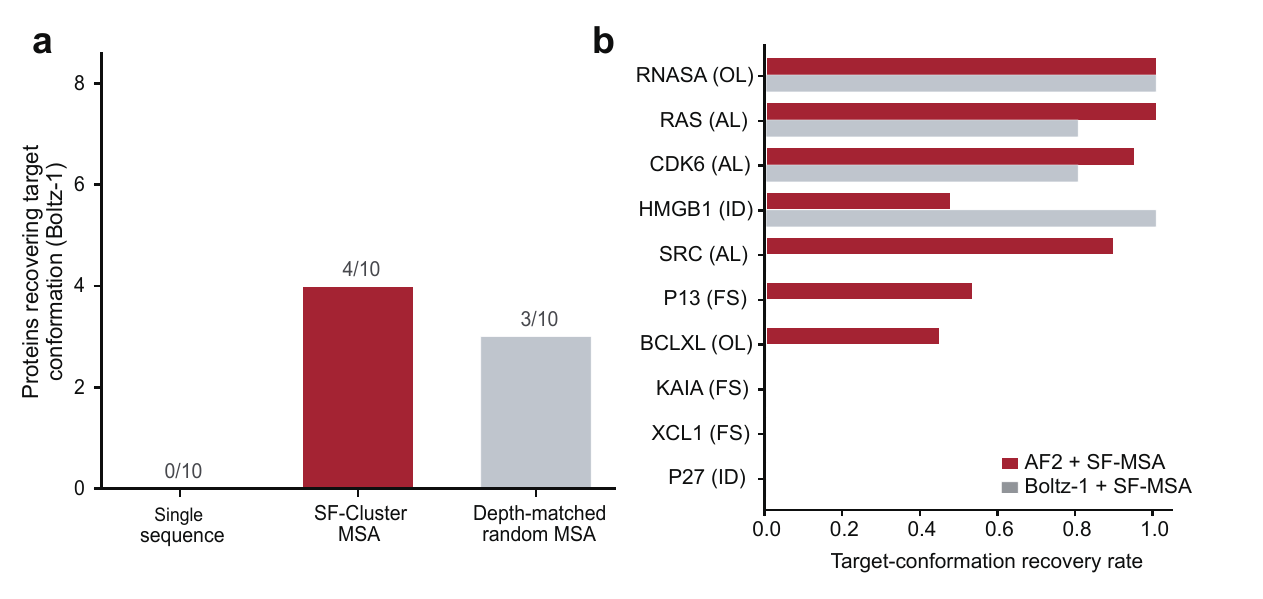}
\caption{\textbf{The selected MSAs carry the conformational signal to a second predictor.} \textbf{(a)} Proteins recovering the target conformation under Boltz-1 ($n=10$): zero from single sequence, four with SF-Cluster MSAs, three with depth-matched random MSAs (SF-Cluster versus depth-matched, McNemar $p=1.0$). \textbf{(b)} Per-protein target recovery under AF2 and Boltz-1 with SF-Cluster MSAs.}
\label{fig:transfer}
\end{figure}

\subsection{Frustration marks functionally and conformationally critical residues}
\label{sec:anchoring}

To understand what information frustration captures, we compared frustration-derived residues against independent experimental measurements of protein function and conformational change. Across both analyses, highly frustrated residues localized to positions associated with functional importance and structural switching.

Against deep mutational scanning data from ProteinGym~\citep{notin2023proteingym}, the residues with the highest frustration scores were enriched at experimentally defined functional sites (Fig.~\ref{fig:representation}a), with a mean enrichment of 3.29-fold across eight assays (median 3.60-fold). Combining per-assay significances yielded a significant overall enrichment (Stouffer $Z = 8.30$, $p = 1.05 \times 10^{-16}$), with significant enrichment observed in five of eight assays, including three SRC datasets~\citep{ahler2019combined,nguyen2023molecular,notin2023proteingym}, DHFR (4.94-fold)~\citep{notin2023proteingym}, and HRAS (3.64-fold)~\citep{bandaru2017deconstruction}. Depth-matched random site sets showed no comparable enrichment. The signal was not universal. No enrichment was observed for calmodulin (0.00-fold)~\citep{weile2017framework} or $\alpha$-synuclein (0.71-fold)~\citep{newberry2020deep}, indicating that frustration preferentially marks functional residues associated with energetic strain rather than functional sites in general.

Frustration also localized to regions that undergo conformational exchange. Across four proteins with NMR two-state exchange measurements, residues identified as switch-confident by SF-Cluster recovered 62.5\% of experimentally observed exchange residues (Fig.~\ref{fig:representation}b). Recall was highest for IscU~\citep{dai2012metamorphic} and MAD2~\citep{luo2004mad2}, intermediate for RfaH~\citep{burmann2012alpha}, and lower for KaiB~\citep{chang2015protein}. Frustration therefore identifies where conformational changes occur, although it does not predict state populations: the correlation between frustration and NMR-derived occupancies was negligible (Pearson $r = 0.06$). Together, these results link frustration to both functional importance and experimentally observed conformational switching.

\begin{figure}[h]
\centering
\includegraphics[width=\linewidth]{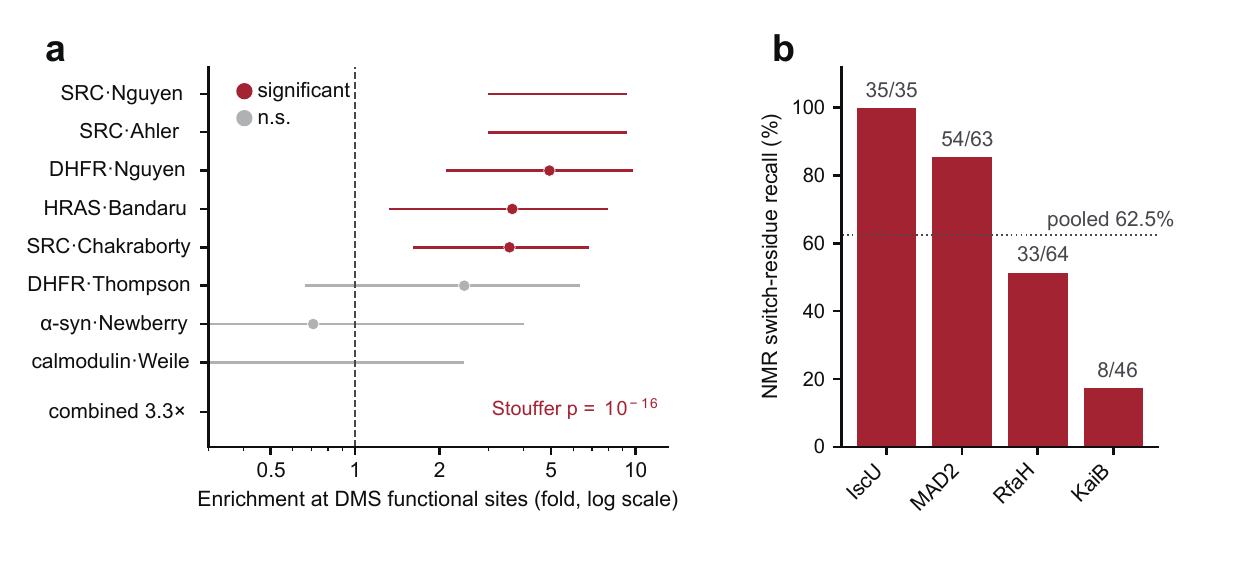}
\caption{\textbf{Frustration marks functional and state-switching residues along an axis distinct from coevolution.} \textbf{(a)} Enrichment of the most-frustrated residues at DMS functional sites across eight assays (red, individually significant; grey, n.s.); mean enrichment $3.3\times$, combined significance Stouffer $p\approx10^{-16}$. \textbf{(b)} Recall of NMR two-state exchange residues by switch-confident frustration, per protein; pooled recall $62.5\%$.}
\label{fig:representation}
\end{figure}

\subsection{Frustration covariation identifies state-switching contacts beyond coevolution}
\label{sec:orthogonality}
We next asked whether frustration covariation captures the same information as classical coevolutionary analyses. Across five proteins, frustration covariation and mean-field direct-coupling analysis shared essentially no variance, with $\rho^2 < 0.001$ in every case and permutation $p \geq 0.15$ (Fig.~\ref{fig:mechanism}a). Frustration covariation and coevolution therefore describe distinct features of homologous sequence variation.

We then examined where each signal localizes within protein structures, scoring enrichment at state-switching contacts, defined as residue-residue interactions that differ between the two reference conformations. Per protein, frustration covariation was enriched at these contacts, significantly so for both KaiB ($p = 0.004$) and RfaH ($p = 1 \times 10^{-4}$), whereas coevolutionary coupling showed no comparable enrichment (Fig.~\ref{fig:mechanism}b). Pooled across proteins, frustration covariation was enriched $1.47$-fold at switching contacts while coevolutionary coupling was depleted there ($0.69$-fold), and neither signal departed appreciably from baseline at the constant contacts shared between the two states ($1.10$-fold and $0.97$-fold; Fig.~\ref{fig:mechanism}c). The distinction was also evident at the level of individual interactions: among the top-ranked frustration-covariation edges, $28$ of the top $100$ in KaiB and $18$ of the top $100$ in RfaH mapped directly to the state-switching interface, where the strongest coevolutionary couplings did not concentrate (Fig.~\ref{fig:mechanism}d,e).

This contact-level signal is also reflected in how the frustration representation selects homologs. For KaiB, AF-Cluster and SF-Cluster draw comparable numbers of sequences from the filtered alignment ($391$ and $394$ of $6{,}821$), but their selections occupy different geometries: in sequence space the two sets overlap the same clusters (Fig.~\ref{fig:mechanism}f), whereas in frustration-pattern space SF-Cluster spreads its picks across the range of per-homolog frustration profiles (Fig.~\ref{fig:mechanism}g). The prediction obtained from this selection adopts the fold-switched (target) conformation rather than the dominant ground-state fold (Fig.~\ref{fig:mechanism}h). Together, these results show that frustration covariation captures information distinct from coevolution, preferentially identifies the contacts that differentiate alternative conformational states, and underlies a selection that recovers the alternative state.
\begin{figure}[h!]
\centering
\includegraphics[width=0.7\linewidth]{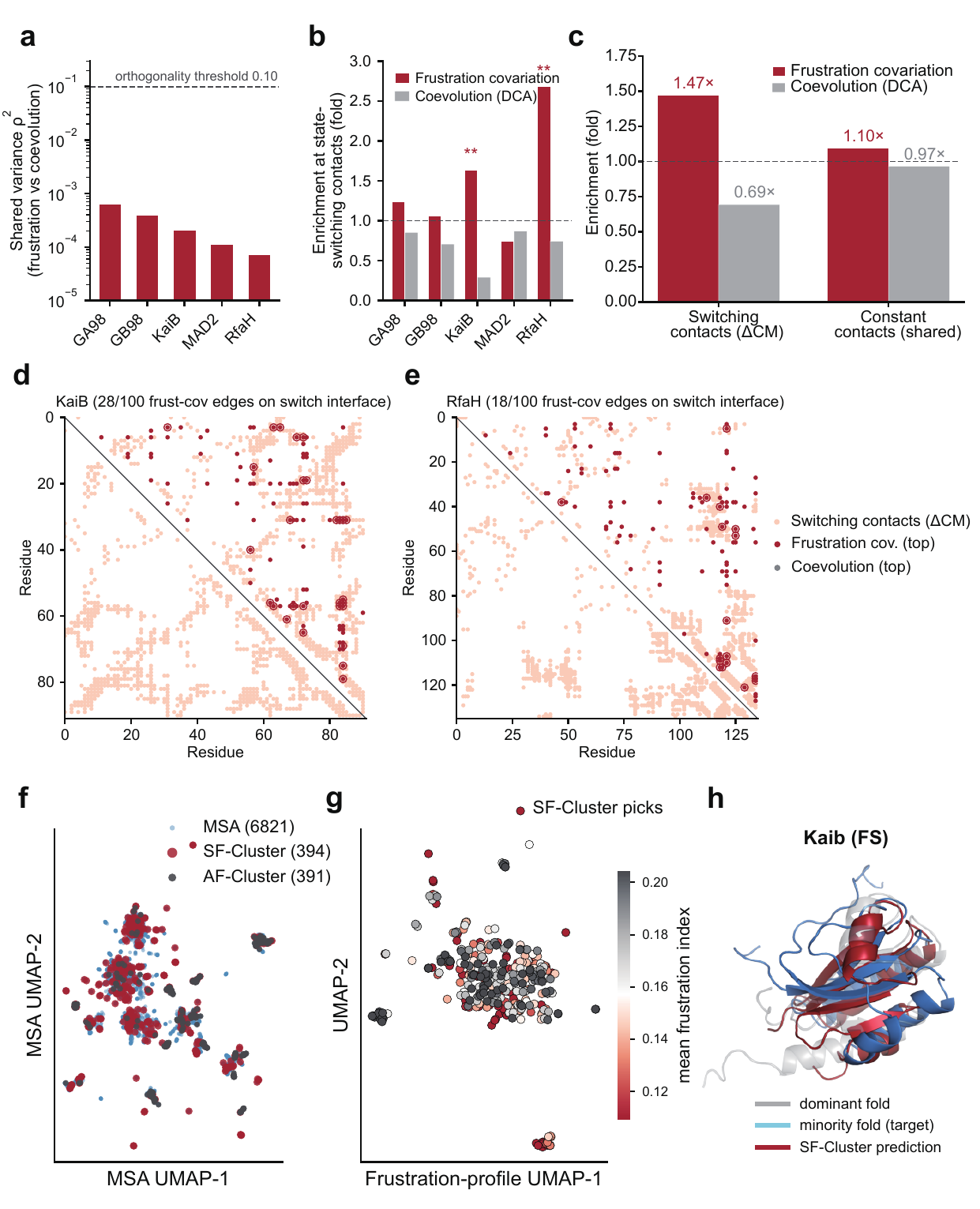}
\caption{\textbf{Frustration covariation targets state-switching contacts that coevolution does not.} \textbf{(a)} Shared variance $\rho^2$ between frustration covariation and coevolution per protein (log scale). \textbf{(b)} Enrichment at state-switching contacts ($\Delta$CM) per protein; significant for frustration on KaiB and RfaH (asterisks). \textbf{(c)} Pooled enrichment at switching versus constant contacts (frustration $1.47\times$ at switching contacts, coevolution $0.69\times$). \textbf{(d,e)} KaiB and RfaH contact maps with top frustration-covariation edges (red) on the switching set ($28/100$ and $18/100$). \textbf{(f--h)} KaiB SF-Cluster picks in sequence space (f) and frustration space (g), and the prediction overlaid on both folds (h).}
\label{fig:mechanism}
\end{figure}

\subsection{SF-Cluster places MSA subsets in a transferable, depth-controlled regime}
\label{sec:dissection}
Finally, we asked what property of the selected MSAs drives the recovery gains of SF-Cluster. Across the benchmark, target-state recovery increased with effective subset depth and saturated near Neff$_{80}\approx30$ for SF-Cluster, depth-matched random, and FI-shuffled subsets, whereas AF-Cluster frequently produced shallower subsets that fell below this regime (Fig.~\ref{fig:depth}a). The recovery difference between SF-Cluster and depth-matched random subsets was centered near zero (Fig.~\ref{fig:depth}b), and aggregate recovery was statistically indistinguishable from the Neff-matched, depth-matched random, FI-shuffled, and coevolution-matched controls while remaining substantially higher than AF-Cluster (Fig.~\ref{fig:depth}c).

The selected subsets differed systematically from those produced by AF-Cluster. SF-Cluster increased mean Neff$_{80}$ from 16.6 to 30.2 (Fig.~\ref{fig:depth}d), reduced within-subset sequence redundancy (Fig.~\ref{fig:depth}e), increased sequence diversity (Fig.~\ref{fig:depth}g), and maintained comparable similarity to the query sequence and comparable switch-region coevolutionary signal (Fig.~\ref{fig:depth}f,h). AF-Cluster instead fragmented large alignments into many small clusters, often preventing individual subsets from reaching the depths associated with successful recovery.

Shuffling frustration indices preserved subset depth almost exactly ($r=0.997$, $n=60$; Fig.~\ref{fig:depth}i) and recovered target conformations at nearly the same rate as SF-Cluster. Effective depth therefore explains most of the aggregate recovery advantage. At matched depth, however, SF-Cluster retained a residual advantage on allosteric proteins and on individual fold-switching systems whose frustration geometry aligned well with the mosaic selection strategy (\Cref{tab:controls}). Frustration-pattern selection thus serves primarily as a mechanism for generating deep, diverse, transferable subsets, while contributing additional state-specific information in a subset of systems.

Together with the residue- and contact-level analyses above, these results place frustration in two roles within SF-Cluster: a selection representation that reliably reaches the depth regime associated with conformational recovery, and a state-resolved descriptor that marks functional residues, conformational exchange regions, and state-switching contacts.
\begin{figure}[!h]
\centering
\includegraphics[width=0.7\linewidth]{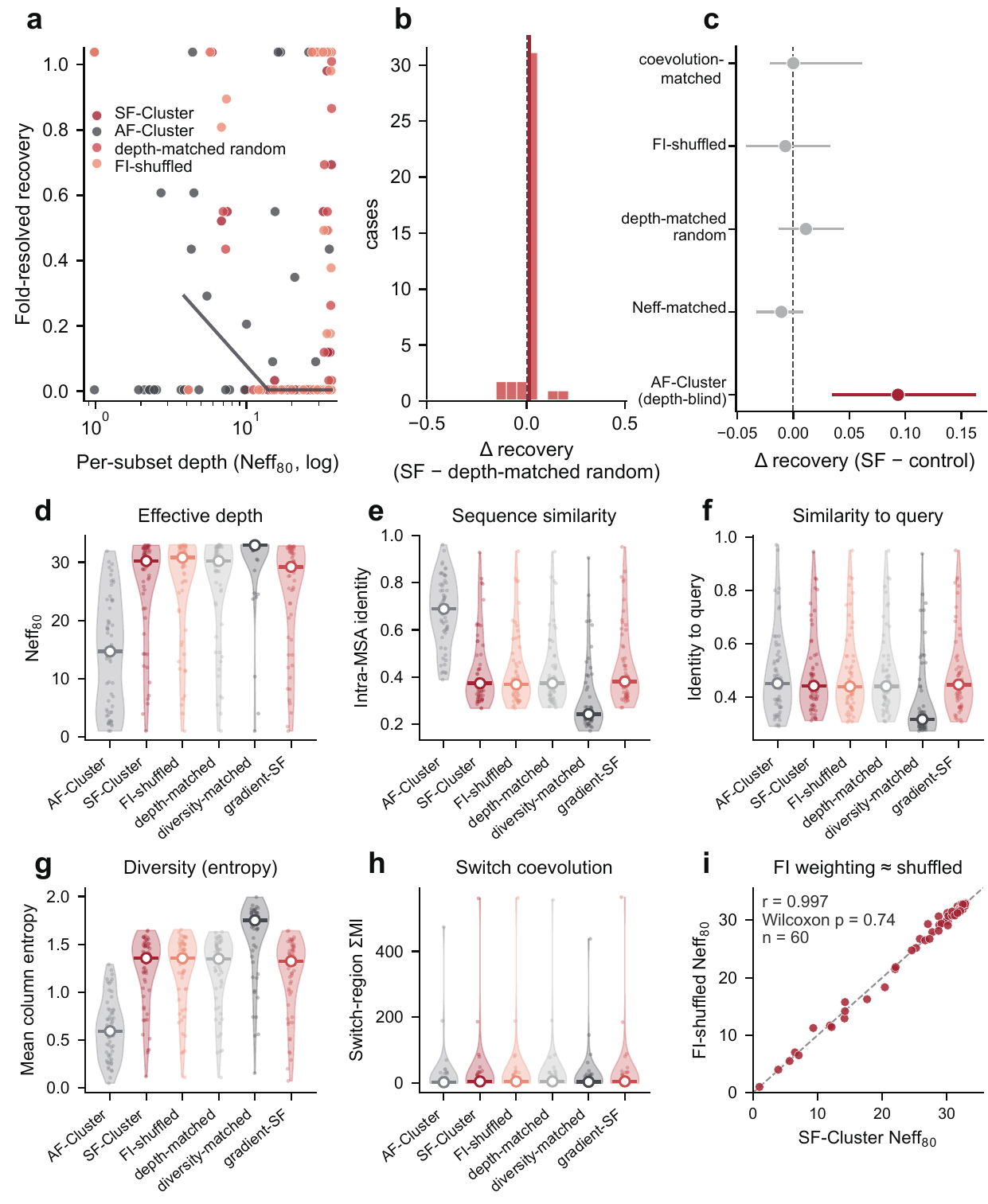}
\caption{\textbf{Recovery is governed by per-subset effective depth.} \textbf{(a)} Target-state recovery (AF-Cluster-style dual-reference endpoint) versus per-subset effective depth (Neff$_{80}$, log scale); recovery rises with effective depth and saturates near Neff$_{80}\approx30$ for SF-Cluster, depth-matched random and FI-shuffled subsets, while depth-blind AF-Cluster subsets fall below this regime. \textbf{(b)} Per-case recovery difference, SF-Cluster minus depth-matched random, centred near zero. \textbf{(c)} SF-Cluster beats AF-Cluster only because AF-Cluster ignores depth: once depth or Neff is controlled for, SF-Cluster matches every baseline, showing the gain comes from effective subset depth, not frustration geometry.
\textbf{(d)} Effective depth at 80\% identity (Neff$_{80}$); \textbf{(e)} mean intra-MSA pairwise identity; \textbf{(f)} mean identity of subset sequences to the query; \textbf{(g)} mean per-column Shannon entropy; \textbf{(h)} switch-region coevolution ($\Sigma$MI), summed mutual information over switch-region column pairs; \textbf{(i)} per-protein Neff$_{80}$ of SF-Cluster versus the FI-shuffled control with the $y=x$ reference line.
}
\label{fig:depth}
\end{figure}

\section{Discussion}

In this work, we introduce SF-Cluster, a frustration-aware MSA subsampling framework that selects sequence subsets in a structure-relevant pattern space rather than by sequence similarity (Fig.~\ref{fig:framework}). Across a benchmark spanning fold-switching, allosteric, oligomerization-coupled, and intrinsically disordered systems, SF-Cluster improves AF-Cluster-style recovery of alternative conformations and transfers this signal to an architecturally distinct predictor through the selected MSA alone (Figs.~\ref{fig:recovery}, \ref{fig:transfer}). Together, these results suggest that MSA subsampling depends less on how sequences are partitioned and more on the representation used to organize the alignment. In this view, MSA subsampling becomes a representation-guided reweighting problem, where the representation determines which conformational states become accessible through subsampling (Theorem~\ref{thm:reweighting}).

Interestingly, matched controls place most of the recovery advantage in effective subset depth rather than in frustration geometry itself. SF-Cluster consistently generates deep, diverse subsets that occupy the regime where recovery saturates, whereas AF-Cluster frequently fragments the same alignment into shallower subsets. At the same time, frustration retains value beyond subset selection. It identifies functional residues, localizes to experimentally observed exchange regions, and enriches for state-switching contacts along an axis distinct from coevolution. We therefore find that frustration serves both as a practical selection representation and as a state-resolved descriptor of conformational variation.

Our Mpt53 benchmark also defines a practical boundary for MSA subsampling. When the underlying sequence pool encodes a single structural basin, neither AF-Cluster nor any SF-Cluster variant recovers an alternative state (Fig.~\ref{fig:si_boundaries}). Subsampling can redistribute probability among conformations already represented in an alignment, but it cannot recover states absent from the sequence pool, consistent with Theorem~\ref{thm:focusing}(b). Coverage-aware refinement addresses a complementary limitation by preserving low-frequency states that would otherwise be discarded during confidence-based ranking.

Looking forward, learned state-aware representations~\citep{leusch2026frustrai, bryant2024structure, lee2025large} may capture conformational signals beyond those available from a single frustration field, while generative MSA augmentation~\citep{zhang2023unsupervisedly, zhang2024msa, chen2024msagpt, cao2025plame, sgarbossa2023generative, venkatraman2025msaflow} and template-guided approaches may introduce structural diversity that no row-subsetting strategy can recover on its own. More broadly, SF-Cluster suggests that the central challenge in MSA subsampling is not generating additional predictions, but selecting informative subsets. We expect this perspective to extend beyond frustration to other learned representations of protein state and dynamics.

\section{Methods}

\subsection{Problem setup}

Let $\mathcal{A} = \{\mathbf{s}_i \in \mathcal{V}^L\}_{i=1}^N$ be a multiple sequence alignment (MSA) of $N$ homologs of length $L$ over an alphabet $\mathcal{V}$. A structure predictor $f$ maps a subset $\mathcal{A}' \subseteq \mathcal{A}$ to a predicted structure $\mathbf{X} = f(\mathcal{A}', \xi)$, where $\xi$ captures stochastic inference, so $f$ induces a distribution $P(\mathbf{X} \mid \mathcal{A}')$ over structure space. We partition structure space into conformational basins $\mathcal{B} = \{B_1,\dots,B_S\}$ with basin probabilities
\begin{equation}
    P(B_s \mid \mathcal{A}') = \int_{B_s} P(\mathbf{X} \mid \mathcal{A}')\, d\mathbf{X}.
\end{equation}
Targeted conformational recovery seeks a small family of subsets $\{\mathcal{A}'_m\}$ that jointly sample a target basin $B_{s^*}$ in addition to the dominant basin, without exhaustively enumerating subsets. This framework describes the two-state case ($S=2$); the single-state recovery task defined below is the degenerate instance in which only one basin is experimentally defined. The formal reweighting analysis, including the focusing and impossibility results invoked in the Discussion, is developed in SI~\ref{si:theory}.

\subsection{Frustration-pattern representation of an MSA}
\label{sec:representation}

Sequence-clustering methods such as AF-Cluster partition $\mathcal{A}$ by sequence similarity, implicitly assuming that the basin $Y_i \in \{1,\dots,S\}$ of homolog $i$ is recoverable from its sequence distances to the other homologs, $P(Y_i \mid \mathcal{A}) \approx P\big(Y_i \mid \{d_{\mathrm{seq}}(\mathbf{s}_i,\mathbf{s}_j)\}_j\big)$. This assumption is weak when a few substitutions reshape fold-switch energetics while leaving global sequence identity nearly unchanged, as in the GB98 T25I/L20A control, where two substitutions ablate the fold-switch signature without substantially altering sequence similarity to wild-type.

We instead represent each homolog by a per-residue frustration profile $F_i \in \mathbb{R}^L$. Local energetic frustration~\citep{ferreiro2007localizing} measures how energetically favourable a residue's interactions are relative to alternative local configurations; residues that are highly or variably frustrated across homologs tend to coincide with hinges, interfaces, and switch regions whose energetic balance differs between states~\citep{freiberger2019local}. Collecting these profiles gives the frustration representation
\begin{equation}
    \Phi(\mathcal{A}) = \{F_i\}_{i=1}^N .
\end{equation}
SF-Cluster selects subsets according to geometric structure in $\Phi(\mathcal{A})$ rather than in sequence space. We compute $F_i$ directly from sequence with FrustrAI-Seq~\citep{leusch2026frustrai}, a previously described sequence-based frustration predictor built on a ProtT5-XL encoder, used here as released without retraining; SF-Cluster uses its per-residue frustration index, whereas the accompanying class label and surprisal score are not used here. Model loading, weight handling, and inference settings are given in SI~\ref{si:impl:frustration}.

\subsection{SF-Cluster: subset selection by frustration-pattern geometry}
\label{sec:pattern_sf}

To make selection tractable we summarise each profile as a region-level embedding. Residues are grouped into $K$ regions $\mathcal{R} = \{R_k\}_{k=1}^K$, and the core descriptor of homolog $i$ is its mean frustration per region,
\begin{equation}
    \psi_0(\mathbf{s}_i) = [\mu_{i1}, \dots, \mu_{iK}], \qquad \mu_{ik} = \frac{1}{|R_k|}\sum_{j \in R_k} F_{ij},
\end{equation}
where $F_{ij}$ is the raw (non-normalised) frustration index at residue $j$ of homolog $i$. The deployed feature vector $\psi(\mathbf{s}_i)$ concatenates these region means with within-region variances and between-region frustration contrasts; exact feature definitions and region boundaries are listed in SI~\ref{si:impl:features}. SF-Cluster selects subsets by a single mosaic operation on $\{\psi(\mathbf{s}_i)\}$: homologs are drawn to span diverse modes of the frustration-pattern distribution, so that each subset carries high within-subset frustration variance. This design is intended to reduce the energetic homogeneity of subsets dominated by the major basin and thereby expose alternative basins to the predictor. This produces $M$ subsets of fixed size, with the query sequence prepended to every subset; subset counts, sizes, and the sampling rule are specified in SI~\ref{si:impl:subsampling}. Throughout, ``SF-Cluster'' refers to this mosaic frustration-pattern selection, applied uniformly to all proteins. Alternative geometric operations on $\psi$ (directional, contrastive, and clustering variants) and residue-level frustration normalisation are evaluated as ablations in SI~\ref{app:strategies}. Representative selections in sequence space and in frustration-pattern space, together with the resulting predictions, are shown for four cases in Fig.~\ref{fig:si_geometry}.

\subsection{Structure prediction and refinement}
\label{sec:refine}

All structures are predicted with AlphaFold2 via the ColabFold~\citep{mirdita2022colabfold} batch interface, using custom-MSA input, with templates, relaxation, and dropout all disabled. Disabling dropout is applied uniformly to every method, including AF-Cluster, so that differences reflect the choice of subsets rather than stochastic ensembling; this is a controlled setting rather than a constraint specific to the baseline. Prediction follows a two-stage screen-and-refine protocol. In the screen stage, every subset is run with one model and two seeds and scored by the mean pLDDT of its highest-confidence structure, and the top $K$ subsets, with $K = \max(4, \min(16, \lceil 0.2\,n_{\text{subsets}} \rceil))$, are promoted to the refine stage, where each is re-run with five models and four seeds. Both stages use three recycling iterations, and refined predictions across promoted subsets are aggregated into the final ensemble. This global pLDDT-ranked refinement is used for all benchmark results reported here. A coverage-aware variant enforces a per-basin quota during promotion to retain low-confidence rare-state predictions; it requires a screen stage deep enough to sample the rare basin, and is described and evaluated in SI~\ref{si:impl:af2} and~\ref{si:impl:refinement} but is not used for the main benchmark. Full prediction and refinement parameters are given in SI~\ref{si:impl:af2}.

\subsection{Experimental setup}
\label{sec:setup}

\paragraph{Benchmarks.} We evaluate SF-Cluster on 48 evaluable cases spanning two distinct prediction tasks. The first is two-state conformational recovery: 33 cases with two experimentally determined reference structures, comprising 11 fold-switching, 12 allosteric and 10 oligomerization-coupled systems, each requiring recovery of a target basin distinct from the dominant one. The fold-switching systems KaiB and the GA98/GB98 pair belong to this set; GB98 T25I/L20A is a designed negative control in which two substitutions ablate the fold-switch signature without substantially altering MSA composition, and is assessed separately (Fig.~\ref{fig:si_boundaries}) rather than counted among the 48. The second is single-state fold recovery for intrinsically disordered regions: 15 cases for which only the bound, ordered conformation is experimentally defined and no second basin exists. The two tasks are scored and reported separately throughout and are never pooled into a single hit rate, and the headline comparison against AF-Cluster is computed on the two-state recovery task. These 48 cases are the evaluable subset of a larger initial assembly of 60 candidate systems; per-protein construction, reference-structure assignment, exclusion of non-evaluable systems, MSA generation, and architecture stratification are described in SI~\ref{si:impl:datasets} and~\ref{si:impl:msa}.

\paragraph{Baselines and controls.} We compare against full-MSA prediction and AF-Cluster~\citep{wayment2024predicting}. To separate the contribution of subset selection from that of subset depth, we add four matched controls: an effective-depth (Neff)-matched control, a depth-matched random control, a frustration-index-shuffled (FI-shuffled) control that preserves subset depth (Fig.~\ref{fig:depth}i) while destroying the frustration pattern, and a coevolution-matched control; their construction is given in SI~\ref{si:impl:controls}. Per-subset MSA properties for all methods and controls, including effective depth, intra-MSA identity and column entropy, are compared in Fig.~\ref{fig:depth}d--i. Across these controls, recovery is governed by effective depth rather than frustration geometry (Fig.~\ref{fig:depth}c), as analysed in \Cref{sec:dissection}. All methods share the same AlphaFold2 backend and screen-and-refine protocol.

\paragraph{Evaluation.} 
For two-state recovery, we used an AF-Cluster-style dual-reference RMSD criterion. Each prediction was independently aligned to the dominant reference state and to the target reference state over the full set of residues shared between the prediction and the corresponding reference. A prediction was counted as a target-state hit only when three conditions were jointly satisfied: (i) the full-chain $\text{C}\alpha$ RMSD to the target reference was at most $3\,\text{\AA}$; (ii) the prediction was closer to the target reference than to the dominant reference; and (iii) the mean pLDDT over the evaluated structure was at least $70$. This dual-reference criterion makes the state assignment mutually exclusive and prevents a prediction that overlaps the target only through a shared structural core from being counted as target-state recovery.

For fold-switching systems, we additionally report switch-region and TM-score-based assignments~\citep{zhang2005tm} as sensitivity analyses. The switch-region analysis asks whether the state-discriminating region is closer to the target than to the dominant state, whereas the TM-score analysis asks whether the global structural assignment is unchanged under a scale-normalised similarity measure. Common-core and loose RMSD criteria are reported only as fold-agnostic structural-proximity upper bounds, not as the primary recovery endpoint. For intrinsically disordered regions, only one ordered reference state exists; these cases are evaluated as single-state ordered-form recovery using RMSD to the bound reference and mean pLDDT, and are reported separately from two-state recovery.

\paragraph{Reproducibility details.} The complete evaluation protocol and per-class hit criteria are specified in SI~\ref{si:impl:eval}; all method and control hyperparameters are consolidated in SI~\ref{si:impl:hparams}; and computational-resource requirements and runtimes are reported in SI~\ref{si:impl:compute}.

\section*{Data and Code Availability}
All data code required to evaluate SF-Cluster is publicly available at \url{https://huggingface.co/ChatterjeeLab/SF-Cluster}. The repository includes an interactive Google Colab notebook for running SF-Cluster on user-defined targets.

\section*{Acknowledgments}
This research was supported by a grant from the High-throughput Institute for Discovery (HIT-ID) at the University of Pennsylvania to the lab of Pranam Chatterjee. The work described in this paper was also partially supported by the Research Grants Council of the Hong Kong Special Administrative Region, China, under Project T45-401/22-N.

\section*{Author Contributions Statement}
H.C., Z.G., and P.C. conceived the study and designed the SF-Cluster approach. H.C. and Z.G. developed the method, implemented the computational pipeline, and carried out all computational experiments and analyses. C.G. and G.L. assisted with a subset of the experiments and contributed to the interpretation and discussion of the results. P.A.H. and P.C. supervised the project. All authors discussed the results and reviewed and approved the final manuscript.

\section*{Competing Interests Statement}
P.C. is a co-founder of Gameto, Inc., UbiquiTx, Inc., AtomBioworks, Inc., and Recognition Bio, Inc., and advises companies involved in biologics and therapeutic development. P.C.'s interests are reviewed and managed by the University of Pennsylvania in accordance with its conflict-of-interest policies. The remaining authors declare no competing interests.

\bibliography{sfcluster}

\clearpage
\appendix
\renewcommand{\thesection}{S\arabic{section}}
\setcounter{section}{0}
\renewcommand{\thesubsection}{S\arabic{section}.\arabic{subsection}}
\renewcommand{\thefigure}{S\arabic{figure}}
\setcounter{figure}{0}
\renewcommand{\thetable}{S\arabic{table}}
\setcounter{table}{0}
\renewcommand{\theequation}{S\arabic{equation}}
\setcounter{equation}{0}

\begin{center}
{\huge\bfseries\color{PennBlue} Supplementary Information}
\end{center}
\vspace{1em}

\section{Implementation Details}
\label{si:impl}

\subsection{Matched-control construction}
\label{si:impl:controls}

All controls share the AlphaFold2 backend, the screen-and-refine protocol and the per-arm subset count and size used by SF-Cluster (12 subsets of 32 sequences; \Cref{si:impl:subsampling}), so that they differ from SF-Cluster only in how sequences are assigned to subsets. The query sequence is prepended to every subset in every control, and all random draws use the fixed seed 0.

\paragraph{Depth-matched random.} For each SF-Cluster subset we record its number of retained sequences and draw a random subset of the same size uniformly without replacement from the filtered MSA. This matches the nominal per-subset depth of SF-Cluster while removing any frustration-based or sequence-based structure in the selection.

\paragraph{Effective-depth (Neff)-matched.} Random subsets are drawn and accepted only if their effective depth, measured as Neff$_{80}$ under the reweighting of \Cref{si:impl:msa}, falls within $\pm10\%$ of the SF-Cluster target ($\mathrm{TOL}=0.10$); up to 50 independent random draws are attempted, after which the closest candidate is retained, with a Hobohm-1 80\%-identity clustering fallback if no draw falls within tolerance. Because uniform random subsets of fixed size and SF-Cluster subsets can differ in redundancy, this control matches the realised information content of the subset rather than its nominal size.

\paragraph{Frustration-index-shuffled (FI-shuffled).} Starting from the SF-Cluster subsets, the per-residue frustration index values are permuted across alignment columns (positions), with a single permutation applied per subset and shared across all homologs rather than drawn independently per homolog, so that amino-acid composition and frustration-index magnitudes are preserved while the position-to-FI correspondence is destroyed, before the region-level features of \Cref{si:impl:features} are recomputed and the mosaic selection is re-run. This destroys the frustration pattern that drives selection while leaving the underlying sequence pool, and therefore the attainable subset depth, unchanged; the resulting subsets have per-case Neff$_{80}$ essentially identical to SF-Cluster (Pearson $r=0.997$, paired Wilcoxon $p=0.74$; median Neff$_{80}$ $30.2$ versus $30.8$ over $n=60$ proteins; Fig.~\ref{fig:depth}i).

\paragraph{Coevolution-matched.} As a control for the contact-level analysis of \Cref{sec:orthogonality}, subsets are selected to match the switch-region coevolutionary signal ($\Sigma$MI; Fig.~\ref{fig:depth}h) of the SF-Cluster subsets by ranking candidate random subsets by their joint closeness in effective depth (Neff$_{80}$) and switch-region coevolution richness (APC-corrected mean mutual information over the top 2\% of column pairs with $|i-j|\ge5$) to the SF-Cluster values, accepting the first candidate whose MI richness and Neff$_{80}$ both fall within $\pm10\%$ of target and otherwise retaining the candidate minimising the summed relative distance (budget: 3{,}000 Neff draws, 200 MI evaluations). This isolates whether recovery is driven by coevolutionary coupling captured in the subset rather than by frustration geometry or depth.

The mean per-case recovery difference between SF-Cluster and each control is reported in \Cref{fig:depth}c, and the corresponding per-class recovery values are tabulated in \Cref{tab:controls}. The coevolution-matched arm could be constructed only on a smaller, different set of matched cases (3 allosteric and 4 oligomerization-coupled cases) and is therefore not row-comparable to the other arms; on those cases its recovery was statistically indistinguishable from SF-Cluster (paired Wilcoxon $p=0.50$), so it is reported here rather than in \Cref{tab:controls}.

\begin{table}[H]
\caption{Per-class recovery on the cleaned benchmark under the AF-Cluster-style dual-reference RMSD endpoint. A two-state prediction is counted as a target-state hit only when its full-chain RMSD to the target reference is at most 3~\AA, its RMSD to the target reference is lower than its RMSD to the dominant reference, and its mean pLDDT is at least 70. SF-Cluster exceeds the depth-blind AF-Cluster in every two-state class, with the clearest gain on allosteric systems. Most of this advantage over AF-Cluster is mediated by effective subset depth: the depth-matched random control reaches broadly comparable recovery, and in aggregate the two are statistically indistinguishable (\Cref{fig:depth}c). The depth account is not complete, however; the SF-Cluster minus depth-matched random column shows a residual advantage of $+6.3$ percentage points on allosteric systems, alongside a near-zero difference on oligomerization-coupled systems and a slightly negative one on fold-switching, so frustration selection contributes beyond depth on some classes but not uniformly.}
\label{tab:controls}
\centering
\small
\setlength{\tabcolsep}{5pt}
\resizebox{\textwidth}{!}{%
\begin{tabular}{lcccccc}
\toprule
Class & $n$ & SF-Cluster & AF-Cluster & Depth-matched random & SF--AF & SF--Depth-random \\
\midrule
Fold-switching & 11 & 0.101 & 0.033 & 0.114 & +0.068 & -0.013 \\
Allosteric & 12 & 0.428 & 0.273 & 0.366 & +0.155 & +0.063 \\
Oligomerization-coupled & 10 & 0.114 & 0.008 & 0.111 & +0.106 & +0.003 \\
\bottomrule
\end{tabular}%
}
\end{table}

\subsection{Pattern-SF strategy comparison and signal attribution}
\label{sec:attribution}
\label{app:strategies}

The main text fixes SF-Cluster to a single mosaic strategy applied uniformly to every protein (\Cref{sec:pattern_sf}). Here we report, as ablations, the per-protein behaviour of the alternative pattern-space strategies and the signal-attribution controls that identify which property of the frustration features carries the effect.

\paragraph{Strategy comparison.} The four pattern-space strategies of \Cref{si:impl:subsampling} (region-cluster, contrast, mosaic, gradient) differ in how they traverse frustration-pattern space. Their recovery on the development fold-switching systems, scored under the looser common-core/switch-region structural-proximity criterion used for these strategy ablations, is given in \Cref{tab:pattern_performance}; these structural-proximity rates are a fold-agnostic upper bound and are not directly comparable to the AF-Cluster-style dual-reference recovery reported in the main text and in \Cref{tab:controls}. The optimal strategy is protein-specific: mosaic is best on KaiB (0.950) and GA98 (0.925), where the conformational signal is distributed across multiple frustration regions, whereas gradient is best on GB98 (0.500), where the signal lies along a single N-to-C directional axis that mosaic dilutes (mosaic 0.188). Allowing the strategy to vary per protein therefore improves recovery on GB98 (gradient 0.500 versus mosaic 0.188) while leaving KaiB and GA98 unchanged, where mosaic is already optimal. Because this improvement is anisotropic and case-specific, we report it here rather than as the headline result, so that no per-protein tuning enters the main comparison against AF-Cluster.

\begin{table}[H]
\caption{Pattern-SF strategy comparison on the development fold-switching systems. Values are common-core/switch-region structural-proximity recovery rates over the $n=80$ refine-stage predictions per condition, a fold-agnostic upper bound used here only to compare strategies and not the AF-Cluster-style dual-reference endpoint of the main text (\Cref{tab:controls}). Mosaic, applied uniformly, is the strategy used in the main text. Mpt53 is a single-basin discovery case and is excluded (\Cref{fig:si_boundaries}).}
\label{tab:pattern_performance}
\centering
\small
\begin{tabular}{lcccc}
\toprule
Protein & Region-cluster & Contrast & Mosaic & Gradient \\
\midrule
KaiB  & 0.538 & 0.675 & 0.950 & 0.650 \\
GA98  & 0.500 & 0.450 & 0.925 & 0.500 \\
GB98  & 0.450 & 0.075 & 0.188 & 0.500 \\
\bottomrule
\end{tabular}
\end{table}

\paragraph{Signal attribution.} To identify which property of the frustration features carries the signal, we compare the raw-FI features used throughout against two perturbations: residue-normalised features, in which the per-column mean FI is subtracted before aggregation (\Cref{si:impl:features}), preserving the spatial pattern while discarding raw amplitudes; and the FI-shuffled selection (\Cref{si:impl:controls}), which destroys the position-to-FI correspondence. Recovery in this attribution analysis is scored under the same common-core/switch-region structural-proximity criterion as the strategy comparison (\Cref{tab:pattern_performance}) and is not directly comparable to the AF-Cluster-style dual-reference recovery of the main text. Residue-normalised features underperform raw features on all three systems (\Cref{tab:attribution}), indicating that the magnitude of the frustration signal, and not only its spatial pattern, contributes to selection. Under the FI-shuffled selection, recovery falls on KaiB (0.950 to 0.537) and GA98 (0.925 to 0.475), indicating that the positional frustration pattern contributes to subset selection on these systems, but rises on GB98 (0.188 to 0.500), where the raw mosaic strategy is already suboptimal (\Cref{tab:pattern_performance}) and removing the positional structure exposes a composition- and depth-driven signal. Because shuffling does not abolish recovery, and on GB98 increases it, we attribute the effect to a combination of positional frustration structure and the underlying sequence composition and depth rather than to positional coupling alone, and we make no claim that the position-to-FI correspondence is necessary for recovery.

\begin{table}[H]
\caption{Signal-attribution controls on the development fold-switching systems. Values are common-core/switch-region structural-proximity recovery rates over the $n=80$ refine-stage predictions per condition (the same fold-agnostic upper bound as \Cref{tab:pattern_performance}, not the AF-Cluster-style dual-reference endpoint of the main text). The raw-FI column is identical to the mosaic column of \Cref{tab:pattern_performance}.}
\label{tab:attribution}
\centering
\small
\begin{tabular}{lccc}
\toprule
Protein & Raw FI (main) & Residue-normalised & FI-shuffled$^{a}$ \\
\midrule
KaiB  & 0.950 & 0.763 & 0.537 \\
GA98  & 0.925 & 0.838 & 0.475 \\
GB98  & 0.188 & 0.050 & 0.500 \\
\bottomrule
\end{tabular}
\end{table}
\noindent\footnotesize $^{a}$ FI-shuffled values are the mosaic-tertile FI-shuffle control: a single per-subset permutation of the per-residue frustration vector across alignment columns, followed by the same Mosaic-SF selection and $n=80$ refine-stage common-core/switch-region structural-proximity scoring as the raw column. \normalsize

\subsection{Frustration index computation}
\label{si:impl:frustration}

Per-residue frustration indices are computed with FrustrAI-Seq~\citep{leusch2026frustrai}, a sequence-based frustration predictor built on a ProtT5-XL encoder. FrustrAI-Seq was trained to predict frustration output from Frustratometer~\citep{jenik2012protein,parra2016protein}. For each homolog in the filtered MSA, the ungapped sequence is passed to the model, which returns a per-residue frustration index (FI), a class label and a surprisal score. Sequences containing non-canonical amino acids are excluded and filled with NaN in all downstream matrices.

The per-residue vectors are projected back to alignment coordinates using each sequence's uppercase-to-column mapping from the A3M format, yielding FI, entropy and surprisal matrices of shape $N \times L$, where $N$ is the number of homologs and $L$ the alignment length; gap positions receive NaN. Inference runs with batch size 1 on a single NVIDIA A100-SXM4-80GB GPU under \texttt{torch.no\_grad()}.

To identify positions of high evolutionary and energetic variation, each alignment column $c$ is scored by
\begin{equation}
    s_c = \operatorname{Var}(\mathrm{FI}_c) \cdot (1 - \overline{\mathrm{entropy}}_c) \cdot (1 + \overline{|\mathbf{z}|}_c),
\end{equation}
where $\operatorname{Var}(\mathrm{FI}_c)$ is the across-sequence variance of FI at column $c$, $\overline{\mathrm{entropy}}_c$ the mean per-sequence entropy and $\overline{|\mathbf{z}|}_c$ the mean absolute surprisal. The top 30 columns by this score are retained as candidate switch columns for the outlier-guided arm.

\subsection{Pattern-level feature representation}
\label{si:impl:features}

Region-level embeddings compress each homolog's FI profile into a fixed-dimensional feature vector using two complementary segmentation schemes. \textbf{Structural segmentation} divides the alignment columns into an N-terminal half (columns 1 to $\lfloor L/2 \rfloor$) and a C-terminal half (columns $\lfloor L/2 \rfloor + 1$ to $L$). \textbf{Variance-based segmentation} designates the top 20\% of columns by across-sequence FI variance as high-variance positions and the remaining 80\% as low-variance positions. Per-sequence features are then computed from these partitions.

\begin{table}[H]
\caption{Per-sequence features derived from the frustration index matrix. All means and variances exclude gap and invalid positions. Feature vectors are standardised (zero mean, unit variance) before clustering.}
\label{tab:si:features}
\centering
\small
\begin{tabular}{ll}
\toprule
Feature & Definition \\
\midrule
$\bar{\mathrm{FI}}_\text{N}$ & Mean FI over N-terminal half \\
$\bar{\mathrm{FI}}_\text{C}$ & Mean FI over C-terminal half \\
$\bar{\mathrm{FI}}_\text{hv}$ & Mean FI over high-variance columns \\
$\bar{\mathrm{FI}}_\text{lv}$ & Mean FI over low-variance columns \\
$\Delta_\text{NC}$ & $\bar{\mathrm{FI}}_\text{N} - \bar{\mathrm{FI}}_\text{C}$ \\
$\Delta_\text{hvlv}$ & $\bar{\mathrm{FI}}_\text{hv} - \bar{\mathrm{FI}}_\text{lv}$ \\
$\sigma^2_\text{N}$ & Within-sequence FI variance over N-terminal half \\
$\sigma^2_\text{C}$ & Within-sequence FI variance over C-terminal half \\
\bottomrule
\end{tabular}
\end{table}

Two FI weightings are used: raw FI values, and residue-normalised FI values in which the per-column mean is subtracted before aggregation. The raw variant is used in all main results; the normalised variant is evaluated as an ablation in \Cref{sec:attribution}.

\subsection{Subsampling strategies}
\label{si:impl:subsampling}

All strategies produce 12 subsets of 32 sequences per arm, with the query sequence prepended to every subset and random seeds fixed at 0 throughout.

\paragraph{AF-Cluster.} The filtered MSA is partitioned by DBSCAN on one-hot encoded sequences over the 21-character alphabet (20 amino acids plus gap), giving feature vectors of dimension $21L$. The $\epsilon$ parameter is selected by scanning $[3.0, 20.0]$ in steps of $0.5$ and applying DBSCAN at each step to a random 25\% subsample, retaining the $\epsilon$ that maximises the number of clusters. Final clustering runs on the full MSA at the selected $\epsilon$ with $\texttt{min\_samples} = 3$, and noise points are discarded.

\paragraph{Arm B (F-Cluster).} Each sequence is embedded by its L2-normalised FI profile with gap positions imputed to 0. $k$-means with $k = 12$ is applied, and within each cluster the 32 sequences nearest the centroid in embedding space are selected.

\paragraph{Arm C (Hybrid-Cluster).} The pairwise clustering distance is
\begin{equation}
    d_{ij} = 0.5\,d_{\text{Hamming}}(i,j) + 0.3\,d_{\cos}^{\mathrm{FI}}(i,j) + 0.2\,d_{\cos}^{\mathrm{entropy}}(i,j),
\end{equation}
where $d_{\text{Hamming}}$ is the normalised pairwise Hamming distance on aligned columns and $d_{\cos}$ is the cosine dissimilarity in the respective feature space. For MSAs with $N > 3{,}000$ sequences, the pairwise Hamming matrix is approximated by concatenating FI and entropy feature vectors weighted by their coefficients and applying $k$-means directly; for smaller MSAs, the distance matrix is embedded into 10-dimensional Euclidean space by classical multidimensional scaling before $k$-means.

\paragraph{Arm D (Outlier-guided).} Each sequence is scored over the top 30 switch columns,
\begin{equation}
    s_i = \sum_{c \in \mathcal{S}} |z_{i,c}| \cdot (1 - \mathrm{entropy}_{i,c}),
\end{equation}
where $z_{i,c}$ is the surprisal at position $c$ and gap entries contribute 0. Sequences are sorted by $s_i$ in descending order and partitioned into 12 bands by round-robin assignment.

\paragraph{Arm E (Diversity-guided).} Twelve subset seeds are chosen by farthest-point sampling in L2-normalised FI space, and each subset comprises the seed and its 31 nearest neighbours in that space. This arm maximises representational diversity across subsets, complementing the compactness of Arm B.

\paragraph{Pattern-SF arms.} Four strategies operate on the region-level feature vectors of \Cref{si:impl:features}:
\begin{itemize}
    \item \textbf{Region-cluster} applies $k$-means ($k = 12$) to the six-dimensional vector $(\bar{\mathrm{FI}}_\text{N},\, \bar{\mathrm{FI}}_\text{C},\, \bar{\mathrm{FI}}_\text{hv},\, \bar{\mathrm{FI}}_\text{lv},\, \sigma^2_\text{N},\, \sigma^2_\text{C})$ after standardisation, then subsamples each cluster to 32 sequences.
    \item \textbf{Contrast} sorts sequences by $\Delta_\text{NC} = \bar{\mathrm{FI}}_\text{N} - \bar{\mathrm{FI}}_\text{C}$ and partitions the sorted list into 12 equal bands, each forming one subset.
    \item \textbf{Mosaic} constructs each subset by strided sampling across the contrast-sorted list, mixing high- and low-$\Delta_\text{NC}$ sequences within each subset to increase N/C frustration variance.
    \item \textbf{Gradient} sorts sequences by $\bar{\mathrm{FI}}_\text{N}$ and partitions into 12 bands, producing a directional sweep from N-frustrated to N-non-frustrated homologs.
\end{itemize}

\subsection{Structure prediction protocol}
\label{si:impl:af2}

All predictions use AlphaFold2 via the ColabFold batch interface (v1.6.1, \texttt{alphafold2\_ptm} model). No structural relaxation is performed, and pLDDT scores are read from the B-factor column of the output PDB. Predictions use custom A3M inputs only, without paired or environmental MSA supplementation.

Prediction follows a two-stage screen-and-refine protocol. In the \textbf{screen stage}, all $n$ subsets for an arm are run with 1 model and 2 seeds, producing 2 structures per subset. In the \textbf{refine stage}, the top $K$ subsets selected by the refinement policy are re-run with 5 models and 4 seeds, producing 20 structures per subset. Both stages use 3 recycling iterations. The number of promoted subsets is
\begin{equation}
    K = \max\!\left(4,\; \min\!\left(16,\; \left\lceil 0.2 \, n_{\text{subsets}} \right\rceil\right)\right).
\end{equation}
The full-MSA baseline uses 5 models and 2 seeds with 3 recycling iterations, producing 10 structures per case.

\subsection{Refinement strategy}
\label{si:impl:refinement}

\paragraph{Global pLDDT selection.} Each subset is represented by its screen-stage structure of highest mean pLDDT, and the top $K$ subsets by this score are promoted to the refine stage.

\paragraph{Coverage-aware refinement for fold-switching cases.} For proteins with known alternative states (KaiB, GA98, GB98), each subset is assigned to the state with the lower switch-region RMSD among its screen structures, and is labelled ambiguous if $|\mathrm{RMSD}_A - \mathrm{RMSD}_B| < 1.0$~\AA. The promotion budget is split equally across non-empty state partitions, with any unfilled quota redistributed to the highest-pLDDT subsets in the remaining pool.

\paragraph{Per-basin refinement for discovery cases.} For Mpt53, screen-stage representative structures are clustered by TM-score using single-linkage at a threshold of 0.80, and within each structural basin the top $\lceil K / n_{\text{basins}} \rceil$ subsets by pLDDT are promoted. This identifies structurally distinct prediction clusters without assuming prior knowledge of alternative states.

\subsection{Protein systems and reference structures}
\label{si:impl:datasets}

\paragraph{Full benchmark composition.} The recovery evaluation comprises 48 evaluable cases spanning four conformational classes: 11 fold-switching (metamorphic) systems, 12 allosteric ligand-induced systems, 10 oligomerization- or domain-swap-coupled systems, and 15 intrinsically disordered regions that undergo disorder-to-order transitions. These 48 cases are the evaluable subset of a larger initial assembly of 60 candidate systems (15 per class); systems lacking a usable two-state or ordered reference were excluded from the recovery evaluation, including two fold-switching entries whose paired references were later found to correspond to different proteins. The 33 two-state systems (fold-switching, allosteric and oligomerization-coupled) each have two experimentally determined reference structures, whereas the 15 disordered-region cases have a single ordered reference and are scored as single-state fold recovery. The dominant/default basin is assigned as the conformation an unguided full-MSA prediction adopts and the target basin as the alternative to be recovered. The full list of evaluable cases, classes and reference structures is given in \Cref{tab:si:full_benchmark}; the six systems detailed in the remainder of this subsection (KaiB, GA98, GB98, Mpt53, RfaH and MAD2) are the development and external-validation set and receive expanded treatment.

\begin{longtable}{llllc}
\caption{The 48 evaluable case benchmark: conformational class and reference structures (11 fold-switching, 12 allosteric, 10 oligomerization-coupled, 15 IDR). State A and State B list the two reference structures (PDB identifiers) for each two-state system; the dominant-versus-target basin assignment for the development systems is given in \Cref{tab:si:proteins}. For the engineered GA/GB designed pair (GA98GB98), State A is the dominant GB fold (2LHD) and State B the target GA fold (2LHC), consistent with \Cref{tab:si:proteins} and the recovery scoring. Disordered-region (IDR) cases carry a single ordered reference and are scored as single-state recovery; IL5 ($^{\dagger}$) is a borderline (weak) oligomerization-coupled case whose two reference entries resolve to the same PDB. Lengths are full experimental construct lengths for the systems also detailed in \Cref{tab:si:proteins} (KaiB, GA98/GB98, Mpt53, RfaH, MAD2); for the remaining systems, for which a separate full-construct length is not defined on disk, the modelled query-region length is given. These 48 evaluable cases are the subset of a 60-system initial assembly used for recovery; per-subset MSA-property analyses (Fig.~\ref{fig:depth}d--i) are computed over the full 60-system assembly.}\label{tab:si:full_benchmark}\\
\toprule
Class & Protein & State A & State B & Length \\
\midrule
\endfirsthead
\multicolumn{5}{l}{\footnotesize\itshape Table \ref{tab:si:full_benchmark} continued}\\
\toprule Class & Protein & State A & State B & Length \\ \midrule \endhead
\midrule \multicolumn{5}{r}{\footnotesize continued on next page}\\ \endfoot
\bottomrule \endlastfoot
Fold-switch & A1AT & 1QLP & 1OPH & 372 \\
Fold-switch & CLIC1 & 1K0M & 1K0N & 235 \\
Fold-switch & GA98GB98 & 2LHD & 2LHC & 56 \\
Fold-switch & IscU & 1WFZ & 2KQK & 130 \\
Fold-switch & KAIA & 1R8J & 1R8Q & 272 \\
Fold-switch & KAIB & 2QKE & 5JYT & 91 \\
Fold-switch & MAD2 & 1DUJ & 1GO4 & 205 \\
Fold-switch & P13 & 2I4Z & 4M1E & 270 \\
Fold-switch & RFAH & 5OND & 6C6S & 162 \\
Fold-switch & SELECASE & 4QHF & 4QHG & 109 \\
Fold-switch & XCL1 & 1J8I & 2JP1 & 93 \\
Allosteric & AKE & 4AKE & 1AKE & 214 \\
Allosteric & CALM & 1CFD & 1CLL & 148 \\
Allosteric & CDK6 & 1BI7 & 2EUF & 269 \\
Allosteric & DHFR & 1RX1 & 1RX4 & 159 \\
Allosteric & GLNBP & 1GGG & 1WDN & 220 \\
Allosteric & LIVBP & 2LIV & 1Z15 & 344 \\
Allosteric & MBP & 1OMP & 1ANF & 370 \\
Allosteric & MEK1 & 3EQI & 3MBL & 315 \\
Allosteric & PHOB & 1B00 & 1ZES & 122 \\
Allosteric & RAS & 4Q21 & 5P21 & 169 \\
Allosteric & RBP & 1URP & 2DRI & 271 \\
Allosteric & SRC & 2SRC & 1Y57 & 449 \\
Oligomer & BARNASE & 1A2P & 1YVS & 108 \\
Oligomer & BCLXL & 1LXL & 1R2D & 221 \\
Oligomer & CALBD9K & 4ICB & 1HT9 & 76 \\
Oligomer & CCL2 & 1DOK & 1DOM & 72 \\
Oligomer & CYSTC & 3GAX & 1TIJ & 107 \\
Oligomer & DTOX & 1F0L & 1DDT & 520 \\
Oligomer & ENGRAILED & 1ENH & 2JWT & 54 \\
Oligomer & IL5$^{\dagger}$ & 1HUL & 1HUL & 108 \\
Oligomer & RNASA & 7RSA & 1A2W & 124 \\
Oligomer & SUC1 & 1PUC & 1SCE & 101 \\
IDR & ASYN & 1XQ8 & -- & 140 \\
IDR & BAD & 1G5J & -- & 25 \\
IDR & BID & 1DDB & -- & 195 \\
IDR & CJUN & 1JNM & -- & 57 \\
IDR & CREB & 1KDX & -- & 27 \\
IDR & EIF4E & 1WKW & -- & 20 \\
IDR & HIF1A & 1L8C & -- & 51 \\
IDR & HISTH1 & 1GHC & -- & 75 \\
IDR & HMGB1 & 1AAB & -- & 83 \\
IDR & NCOA1 & 1KBH & -- & 59 \\
IDR & NUMB & 1WJ1 & -- & 156 \\
IDR & P27 & 1JSU & -- & 69 \\
IDR & P53TAD & 2K8F & -- & 90 \\
IDR & PKID & 1KDX & -- & 27 \\
IDR & PUMA & 2M04 & -- & 25 \\
\end{longtable}

The systems below receive detailed treatment as development and external-validation cases. KaiB, GA98, GB98 and Mpt53 constitute the development set; RfaH and MAD2 are held out for external validation. GB98 T25I/L20A is a designed negative control derived from the GB98 sequence and is not used for method development.

\begin{table}[H]
\caption{Reference structures for the development and validation systems. State A denotes the dominant or native conformation; State B denotes the fold-switched or target alternative. Construct lengths are the residue ranges used for all predictions and evaluations. For the engineered GA98/GB98 pair the GB fold (2LHD) is the basin an unguided full-MSA prediction adopts (State A) and the GA fold (2LHC) is the targeted alternative (State B) for both proteins, consistent with the recovery scoring and the S2 pipeline. Among these systems, KaiB, the GA98/GB98 pair, RfaH and MAD2 are part of the 48-case recovery benchmark (\Cref{tab:si:full_benchmark}); Mpt53 is a single-basin discovery case and GB98 T25I/L20A a designed negative control, both assessed separately (\Cref{fig:si_boundaries}) and not counted among the 48 evaluable cases.}
\label{tab:si:proteins}
\centering
\small
\setlength{\tabcolsep}{4pt}
\begin{tabular}{lllll}
\toprule
Protein & Length & State A & State B & UniProt \\
\midrule
KaiB           & 91 aa  & 2QKE (chain B) & 5JYT (chain A) & Q79V61 \\
GA98           & 56 aa  & 2LHD (chain A) & 2LHC (chain A) & - \\
GB98           & 56 aa  & 2LHD (chain A) & 2LHC (chain A) & - \\
GB98 T25I/L20A & 56 aa  & 2LHE (chain A) & -            & - \\
Mpt53          & 136 aa & 1LU4 (chain A) & discovery task & P9WG65 \\
RfaH           & 162 aa & 5OND (chain A) & 6C6S (chain D) & P0AFW0 \\
MAD2           & 205 aa & 1DUJ (chain A) & 1GO4 (chain A) & Q13257 \\
\bottomrule
\end{tabular}
\end{table}

\paragraph{KaiB.} The construct spans residues 5--95 of PDB 2QKE chain B, trimmed from the 108-residue crystal chain to remove disordered terminal residues. State B (5JYT) carries the stabilising mutations Y8A, N29A, G89A, D91R and Y94A relative to wild-type and was determined by NMR as an ensemble of 20 models; the first model is used as the reference.

\paragraph{GA and GB proteins.} The GA and GB families are combinatorially engineered 56-residue proteins that adopt distinct $\alpha$-helical (GA) and $\beta$-sheet (GB) folds. The primary fold-switch pair evaluated is GB98 versus GA98. For both queries the gap-filtered alignment is dominated by GB-fold ($\beta$) homologs, so an unguided full-MSA prediction defaults to the GB fold; we therefore take the GB fold (2LHD) as the dominant/default basin (State A) and the GA fold (2LHC) as the targeted alternative basin (State B) for both proteins, and recovery is scored against the GA fold (2LHC) for both (\Cref{tab:si:proteins}). Structures were retrieved from the RCSB and trimmed to residues 1--56, with no further trimming required.

\paragraph{Mpt53.} The construct spans residues 38--173 in UniProt coordinates (P9WG65), corresponding to crystal residues 1001--1134 of PDB 1LU4 chain A with offset 963 applied. No experimentally confirmed alternative fold is known for Mpt53, so it is treated as a discovery benchmark whose objective is to identify non-native basins rather than recover a known target.

\paragraph{RfaH.} Full-length RfaH (residues 1--162) is used. State A is the NusG-like autoinhibited conformation (5OND) and State B is the fold-switched C-terminal-domain conformation observed in the RNAP elongation complex (6C6S). The fold-switch region is the C-terminal domain, approximately residues 101--162.

\paragraph{MAD2.} The full 205-residue construct is used. State A is the open O-MAD2 conformation (1DUJ) and State B is the closed C-MAD2 conformation in complex with MAD1 (1GO4). The switch involves topological rearrangement of the safety-belt region, approximately residues 170--205.

\paragraph{GB98 T25I/L20A.} This double mutant of GB98 (2LHE) is used exclusively as a negative control. The two substitutions occur at positions implicated in the fold-switch interface and together ablate the fold-switch signal detectable in the MSA. It is not used for method development or tuning.

\subsection{MSA construction}
\label{si:impl:msa}

Multiple sequence alignments were generated with ColabFold v1.6.1 using the \texttt{--msa-only} flag. The MMseqs2 search ran against the combined UniRef30 and environmental databases via the public ColabFold API server, with all queries submitted as single FASTA sequences without template search; alignments were retrieved in A3M format.

Raw A3M files were filtered following \citet{wayment2024predicting}: a non-query sequence is retained if and only if its fraction of gap characters over aligned columns is at most 25\%, with lowercase A3M insertion characters excluded from the gap count. The query sequence is always retained.

\begin{table}[H]
\caption{MSA statistics after gap filtering for the development and validation systems. Neff is computed at a sequence identity threshold of 0.8 using the standard reweighting formula; it is not reported for the validation-only cases RfaH and MAD2.}
\label{tab:si:msa}
\centering
\small
\setlength{\tabcolsep}{4pt}
\begin{tabular}{llcccc}
\toprule
Protein & Variant & Raw depth & Filtered & Length & Neff \\
\midrule
KaiB  & KaiB            & 7,895 & 6,821 & 91  & 4,570 \\
GA/GB & GA98            &   438 &   179 & 56  & 22.1  \\
      & GB98            &   396 &   181 & 56  & 26.1  \\
      & GB98 T25I/L20A  &   369 &   169 & 56  & 26.6  \\
Mpt53 & Mpt53           & 8,418 & 5,647 & 136 & 4,801 \\
\midrule
RfaH  & RfaH            & 5,606 & 3,359 & 162 & -   \\
MAD2  & MAD2            & 3,179 & 2,091 & 205 & -   \\
\bottomrule
\end{tabular}
\end{table}

\subsection{Evaluation metrics}
\label{si:impl:eval}

\paragraph{Reference pairing and evaluated residues.} For every two-state protein, we define a dominant reference state $D$ and a target reference state $T$ before evaluating any predictions. The dominant state is the state adopted by the unguided full-MSA prediction, and the target state is the alternative state to be recovered. For each prediction $X$, structural comparisons are computed separately against $D$ and $T$. The evaluated residue set for each comparison is the set of residues that can be mapped between the prediction and the corresponding reference after applying the benchmark residue mapping. Cases with unresolved chain or residue mapping are excluded from recovery statistics and listed as non-evaluable.

\paragraph{Primary AF-Cluster-style dual-reference endpoint.} The primary two-state recovery metric follows the dual-reference logic used by AF-Cluster: a prediction must be structurally close to the target reference and closer to the target than to the dominant reference. Let
\[
R_T(X) = \mathrm{RMSD}(X,T), \qquad R_D(X) = \mathrm{RMSD}(X,D),
\]
where RMSD is the C$\alpha$ RMSD after optimal superposition to the corresponding reference. A prediction is counted as a target-state hit if and only if
\[
R_T(X) \leq 3.0~\text{\AA}, \qquad R_T(X) < R_D(X), \qquad \overline{\mathrm{pLDDT}}(X) \geq 70.
\]
The same criterion is applied to all two-state methods and controls. This endpoint is mutually exclusive: a prediction cannot simultaneously be counted as both dominant-state and target-state recovery.

\paragraph{State assignment.} Predictions are assigned to the target state, dominant state, ambiguous state, or low-confidence state using the same dual-reference quantities. Predictions satisfying the target-hit criterion above are labelled target-state hits. Predictions with $R_D(X) \leq 3.0~\text{\AA}$, $R_D(X) < R_T(X)$, and $\overline{\mathrm{pLDDT}}(X) \geq 70$ are labelled dominant-state predictions. Predictions with mean pLDDT below 70 are labelled low-confidence. Predictions that are not within 3.0~\AA{} of either state, or that are within the threshold but do not unambiguously favour one reference over the other, are labelled ambiguous and are not counted as target-state hits.

\paragraph{Sensitivity analyses.} We report three sensitivity analyses. First, we compute a margin-based dual-reference endpoint requiring
\[
R_T(X) \leq 3.0~\text{\AA}, \qquad R_T(X)+1.0~\text{\AA} \leq R_D(X), \qquad \overline{\mathrm{pLDDT}}(X) \geq 70.
\]
Second, for fold-switching proteins with predefined switch regions, we compute switch-region RMSD to both references and ask whether the switch region is closer to the target than to the dominant state. This analysis is used as a state-discriminating sensitivity check rather than as the primary hit definition. Third, we compute TM-score to both references as a scale-normalised assignment check. Common-core and loose RMSD criteria are reported only as fold-agnostic structural-proximity upper bounds and are not used as the primary recovery metric.

\paragraph{Single-state IDR evaluation.} Intrinsically disordered-region cases have only one experimentally defined ordered reference and are not two-state recovery tasks. For these cases, a prediction is counted as an ordered-form hit when the C$\alpha$ RMSD to the ordered reference is at most 3.0~\AA{} and the mean pLDDT is at least 70. IDR results are therefore reported separately from two-state conformational recovery.

\paragraph{Non-evaluable assembly-coupled cases.} For oligomerization- or domain-swap-coupled systems, recovery is evaluated only when the available prediction contains the structural unit required to distinguish the two states. Cases in which the experimentally defined target state depends on an assembly that is absent from the prediction are marked as non-evaluable for assembly-state recovery and are not interpreted as successful or failed target-state recovery.

\subsection{Hyperparameters}
\label{si:impl:hparams}

All hyperparameters are fixed before any AlphaFold2 inference. The table below lists every parameter that materially affects results.

\begin{table}[H]
\caption{Hyperparameter settings. All values are fixed prior to inference and held constant across proteins unless noted.}
\label{tab:si:hparams}
\centering
\resizebox{0.8\columnwidth}{!}{
\begin{tabular}{llc}
\toprule
Component & Parameter & Value \\
\midrule
MSA filtering      & Gap fraction threshold          & 25\% \\
\midrule
AF-Cluster         & Alphabet size                   & 21 \\
                   & \texttt{min\_samples}           & 3 \\
                   & $\epsilon$ scan range           & $[3.0, 20.0]$, step 0.5 \\
                   & $\epsilon$ scan subsample       & 25\% \\
\midrule
Subsampling        & Subsets per arm                 & 12 \\
                   & Sequences per subset            & 32 \\
                   & $k$-means clusters              & 12 \\
                   & Switch columns retained (Arm D) & 30 \\
                   & Arm C weights $(\alpha,\beta,\gamma)$ & $(0.5, 0.3, 0.2)$ \\
                   & Arm C MDS dimensionality        & 10 \\
                   & Large-MSA threshold (Arm C)     & $N > 3{,}000$ \\
                   & High-variance percentile        & 80th \\
\midrule
Controls           & Neff$_{80}$ match tolerance     & $\pm10\%$ \\
                   & Max random draws (Neff-match)   & 50 \\
                   & Coev MI tolerance               & $\pm10\%$ \\
                   & Coev draw / MI-eval budget      & 3{,}000 / 200 \\
\midrule
AlphaFold2         & Model type                      & \texttt{alphafold2\_ptm} \\
                   & Screen: models $\times$ seeds   & $1 \times 2$ \\
                   & Refine: models $\times$ seeds   & $5 \times 4$ \\
                   & Recycling iterations            & 3 \\
                   & Random seed                     & 0 \\
\midrule
Refinement         & $K$ (refine budget)             & $\max(4, \min(16, \lceil 0.2n \rceil))$ \\
                   & Ambiguity threshold             & 1.0 \AA \\
                   & Basin TM-score linkage          & 0.80 \\
\midrule
Evaluation         & Dual-reference RMSD threshold (primary) & 3.0 \AA \\
                   & Target-state assignment              & $R_T < R_D$ \\
                   & Hit pLDDT threshold (overall)        & 70 \\
                   & Hit pLDDT threshold (switch region)  & 70 \\
                   & Collapse-score TM threshold          & 0.80 \\
                   & Common-core proximity threshold (sensitivity) & 3.0 \AA \\
                   & Switch-region displacement threshold (sensitivity) & 3.0 \AA \\
\bottomrule
\end{tabular}
}
\end{table}

\subsection{Computational resources}
\label{si:impl:compute}

All experiments ran on a single server with 8 $\times$ NVIDIA A100-SXM4-80GB GPUs, with FrustrAI-Seq inference and ColabFold predictions each allocated an exclusive GPU per job.

FrustrAI-Seq inference times per case were: KaiB (6,816 sequences, 91 aa) 24 min; Mpt53 (5,638 sequences, 136 aa) 20 min; RfaH (3,350 sequences, 162 aa) 2 min; MAD2 (2,066 sequences, 205 aa) 1 min; and the GA/GB variants (167--177 sequences, 56 aa) 20--44 s each.

For AlphaFold2 inference, the screen and refine stages required approximately 5 and 50 minutes per arm for short proteins (56 aa) and approximately 10 and 80 minutes per arm for longer proteins (91--136 aa). Across all experiments approximately 17,000 AlphaFold2 model evaluations were performed. Peak GPU memory per job was 12--25 GB for proteins up to 162 residues. AF-Cluster on RfaH exceeded available VRAM owing to the large number of DBSCAN clusters (116) and was aborted; this failure is noted here rather than in a separate table.

\section{Theoretical Analysis}
\label{si:theory}

\subsection{Setup}

Let $\mathcal{A}$ be a full MSA and $\mathcal{A}'\subseteq\mathcal{A}$ a selected sub-MSA. A structure predictor induces a conditional distribution $P(\mathbf{X}\mid\mathcal{A}')$ over structure space, partitioned into conformational basins $\mathcal{B}=\{B_1,\dots,B_S\}$ with basin probabilities
\[
    P(B_s\mid\mathcal{A}')=\int_{B_s}P(\mathbf{X}\mid\mathcal{A}')\,d\mathbf{X}.
\]
We assume the predictor depends on the MSA through a representation $\phi(\mathcal{A}')$, and that basin probabilities take the softmax form
\[
    P(B_s\mid\mathcal{A}')= \frac{\exp(g_s(\phi(\mathcal{A}')))}{\sum_{r=1}^{S}\exp(g_r(\phi(\mathcal{A}')))},
\]
where $g_s(\cdot)$ is an unknown basin-specific compatibility score.

\subsection{Theorem 1: MSA subsampling induces conformational reweighting}

\begin{theorem}
\label{thm:reweighting}
Suppose there exist two sub-MSAs $\mathcal{A}'_1,\mathcal{A}'_2\subseteq\mathcal{A}$ with $\phi(\mathcal{A}'_1)\neq\phi(\mathcal{A}'_2)$, and that for some basin $B_s$ and at least one $r\neq s$,
\[
    \bigl[g_s(\phi(\mathcal{A}'_1))-g_r(\phi(\mathcal{A}'_1))\bigr] \;\neq\; \bigl[g_s(\phi(\mathcal{A}'_2))-g_r(\phi(\mathcal{A}'_2))\bigr].
\]
Then $P(B_s\mid\mathcal{A}'_1)\neq P(B_s\mid\mathcal{A}'_2)$.
\end{theorem}

\begin{proof}
The log-odds between basins $s$ and $r$ under any sub-MSA $\mathcal{A}'$ is
\[
    \log\frac{P(B_s\mid\mathcal{A}')}{P(B_r\mid\mathcal{A}')} = g_s(\phi(\mathcal{A}')) - g_r(\phi(\mathcal{A}')).
\]
By assumption this quantity differs between $\mathcal{A}'_1$ and $\mathcal{A}'_2$ for at least one $r\neq s$, so the relative odds of $B_s$ versus $B_r$ change between the two sub-MSAs, which implies $P(B_s\mid\mathcal{A}'_1)\neq P(B_s\mid\mathcal{A}'_2)$.
\end{proof}

\noindent Theorem~\ref{thm:reweighting} establishes that MSA subsampling is a conformational reweighting operation whose effect is mediated entirely by the shift in representation $\phi(\mathcal{A}')$, making the choice of $\phi$ the binding constraint on what subsampling can achieve.

\subsection{Theorem 2: Conditions for focusing and impossibility}

\begin{theorem}
\label{thm:focusing}
\textbf{(a) Targeted focusing.} Suppose there exist a target basin $B_{s^*}$ and a sub-MSA $\mathcal{A}'$ such that
\[
    g_{s^*}(\phi(\mathcal{A}')) - g_s(\phi(\mathcal{A}')) \;\geq\; \Delta \quad \forall\, s\neq s^*.
\]
Then $P(B_{s^*}\mid\mathcal{A}') \;\geq\; \frac{1}{1+(S-1)e^{-\Delta}}$, and $P(B_{s^*}\mid\mathcal{A}')\to 1$ as $\Delta\to\infty$.

\medskip\noindent\textbf{(b) Single-basin impossibility.} Suppose there exists a dominant basin $B_{s_0}$ such that for every $\mathcal{A}'\subseteq\mathcal{A}$,
\[
    g_{s_0}(\phi(\mathcal{A}')) - g_s(\phi(\mathcal{A}')) \;\geq\; \Delta \quad \forall\, s\neq s_0.
\]
Then $P(B_{s_0}\mid\mathcal{A}') \;\geq\; \frac{1}{1+(S-1)e^{-\Delta}}$ for every sub-MSA $\mathcal{A}'$, and no row-subsetting strategy can recover an alternative basin with high probability when $\Delta$ is large.
\end{theorem}

\begin{proof}[Proof of Part (a)]
Let $z_s = g_s(\phi(\mathcal{A}'))$. By assumption $z_s \leq z_{s^*} - \Delta$ for all $s\neq s^*$, so $\sum_{s\neq s^*} e^{z_s} \leq (S-1)e^{z_{s^*}}e^{-\Delta}$. Substituting into the softmax denominator,
\[
    P(B_{s^*}\mid\mathcal{A}') = \frac{e^{z_{s^*}}}{e^{z_{s^*}}+\sum_{s\neq s^*}e^{z_s}} \;\geq\; \frac{e^{z_{s^*}}}{e^{z_{s^*}}+(S-1)e^{z_{s^*}}e^{-\Delta}} = \frac{1}{1+(S-1)e^{-\Delta}}.
\]
As $\Delta\to\infty$, $e^{-\Delta}\to 0$, giving $P(B_{s^*}\mid\mathcal{A}')\to 1$.
\end{proof}

\begin{proof}[Proof of Part (b)]
The argument is identical to Part (a) with $s^*$ replaced by $s_0$, except that the margin condition holds uniformly over all $\mathcal{A}'\subseteq\mathcal{A}$. The same bound therefore gives $P(B_{s_0}\mid\mathcal{A}') \geq \frac{1}{1+(S-1)e^{-\Delta}}$ for every feasible subset, so no subsampling strategy can meaningfully shift probability mass away from $B_{s_0}$.
\end{proof}

\paragraph{Implication.} Part (b) formalizes the empirical observation that SF-Cluster fails to recover an alternative conformation for Mpt53: when the MSA pool is structurally single-basin, every feasible $\phi(\mathcal{A}')$ remains dominated by $g_{s_0}$, and subsampling cannot create conformations absent from the input pool.

\section{Supplement Figures}
\setcounter{figure}{0}
\renewcommand{\figurename}{Extended Data Fig.}

\begin{figure}[h]
\centering
\includegraphics[width=\linewidth]{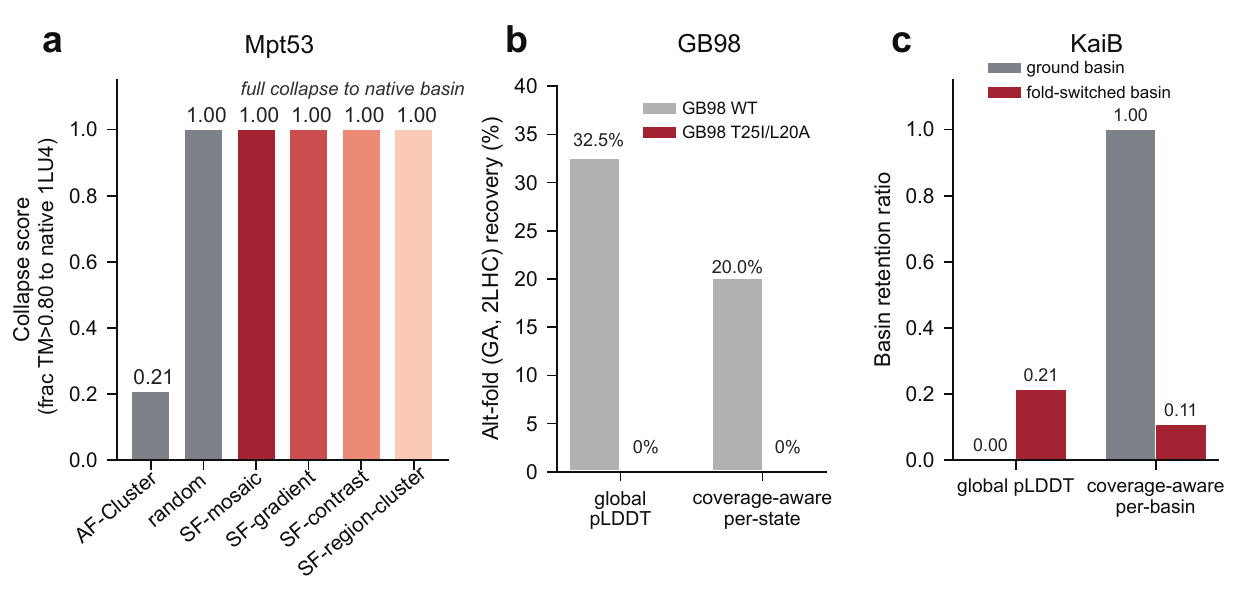}
\caption{\textbf{Boundaries of frustration-guided subsampling.} The framework defines its own limits across three regimes. \textbf{(a)} Mpt53 discovery case: no subsampling method escapes the single native basin. Collapse score is the fraction of predictions within TM-score $0.80$ of the native fold (1LU4); the SF and random arms sit at full collapse ($1.00$) and AF-Cluster, despite a lower collapse ($0.21$), still recovers no alternative basin, consistent with the single-basin impossibility of Theorem~\ref{thm:focusing}(b). \textbf{(b)} GB98 designed pair: shared-core proximity is not global fold recovery. Under the common-core structural-proximity criterion (used for the fold-switch ablations, not the AF-Cluster-style dual-reference endpoint of the main text), $32.5\%$ of wild-type GB98 refine-stage predictions under global-pLDDT refinement and $20.0\%$ under coverage-aware per-state refinement fall within $3$~\AA{} of the GA reference (2LHC) on the shared structural core. Under the full-chain dual-reference endpoint, however, recovery of the GA fold is $0\%$ under both refinement policies: these predictions remain in the dominant GB fold (minimum full-chain C$\alpha$ RMSD $1.49$~\AA{} to 2LHD) and never globally adopt the GA fold (minimum full-chain RMSD $3.7$~\AA, beyond the $3$~\AA{} threshold). The common-core hits are therefore satisfied by the $\beta$-core shared between the two folds while the global fold is unchanged, illustrating why the dual-reference endpoint is used as the primary criterion. The designed T25I/L20A double mutant, which destroys the fold-switch signal in the alignment, reaches $0\%$ under the common-core criterion as well, removing even the shared-core proximity. \textbf{(c)} KaiB refinement: a basin-coverage trade-off. Global-pLDDT ranking retains the well-sampled fold-switched basin but drops the under-sampled ground basin to zero retention ($0/6$), whereas the coverage-aware per-basin quota rescues the ground basin ($6/6$) at a modest cost to fold-switched retention ($16/75 \to 8/75$).}
\label{fig:si_boundaries}
\end{figure}

\begin{figure}[h]
\centering
\includegraphics[width=0.7\linewidth]{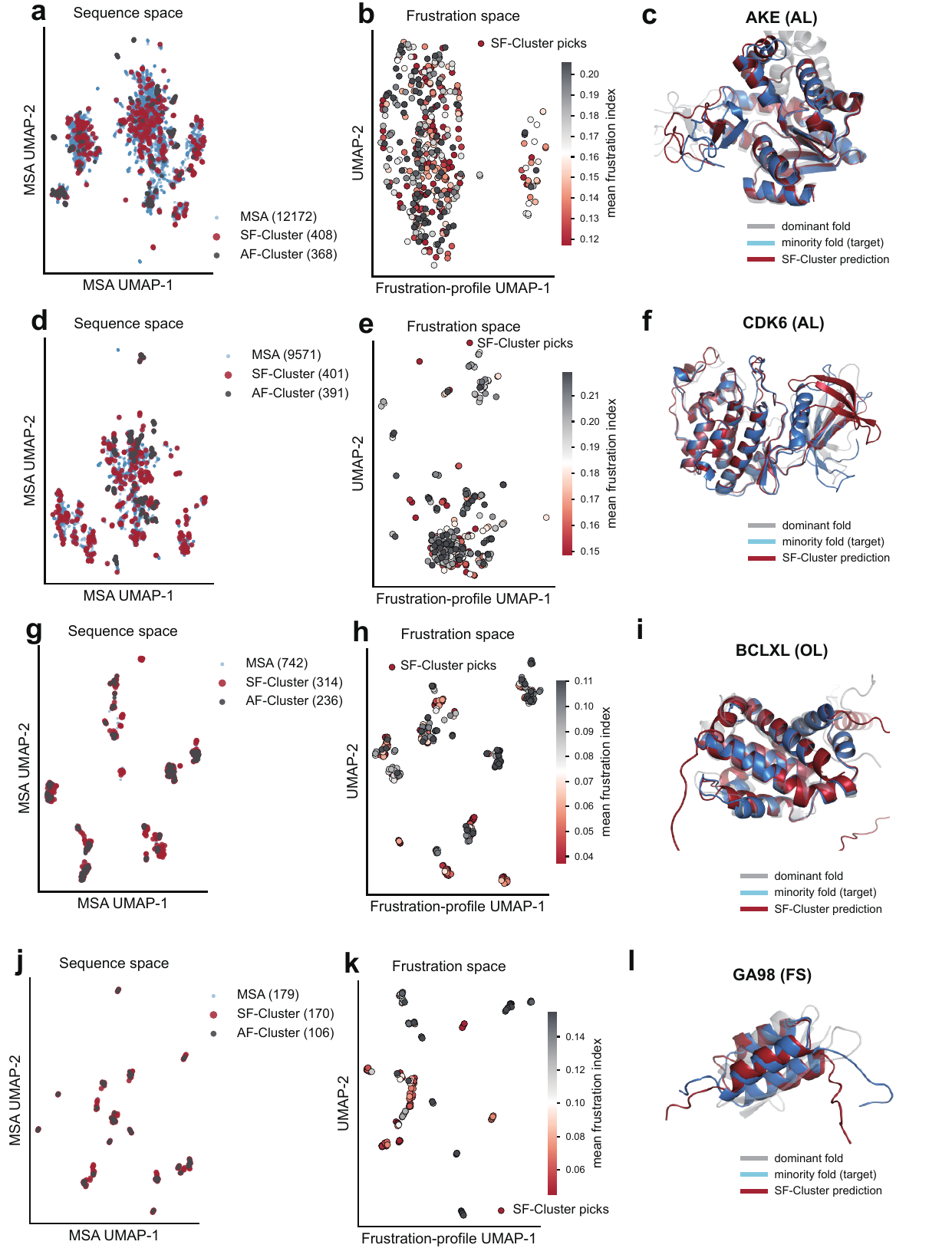}
\caption{\textbf{Selection geometry and structural-proximity examples for four representative cases from the benchmark.} Rows: AKE (alternate-loop, AL), CDK6 (AL), BCL-XL (occluded-loop, OL) and GA98 (fold-switch, FS); membership in the benchmark was asserted before plotting. \textbf{Left}, sequence-space UMAP of the full filtered MSA (gray cloud) with AF-Cluster picks (gray) and SF-Cluster picks (red) overlaid and subset counts annotated. \textbf{Middle}, UMAP of per-homolog frustration-index profiles, each homolog coloured by its mean frustration index on a sequential red ramp (5th--95th-percentile clip), with SF-Cluster picks outlined. \textbf{Right}, C$\alpha$ superposition of a representative SF-Cluster screen prediction, chosen to illustrate structural proximity to the target reference (the prediction of lowest full-chain RMSD to the target among the displayed subset), shown over the dominant fold (gray, transparent) and the target/minority fold reference (blue), with the SF-Cluster prediction in red. These panels illustrate structural proximity only; not all displayed predictions satisfy the AF-Cluster-style dual-reference target-hit criterion. UMAP parameters: seed $20260422$, \texttt{n\_neighbors}$=\min(15,\max(5,N/50))$, \texttt{min\_dist}$=0.3$.}
\label{fig:si_geometry}
\end{figure}

\end{document}